\documentclass[11pt]{article}
\usepackage[letterpaper,textwidth=6.5in,textheight=9in,
            centering,ignorehead,nomarginpar]{geometry}
\usepackage[T1]{fontenc}
\usepackage{textcomp}
\usepackage{xcolor}

\usepackage{url}
\usepackage{hyperref}
\usepackage{doi}

\usepackage{array}
\usepackage{natbib}
\bibpunct{(}{)}{;}{a}{,}{;}
\usepackage{graphicx}
\usepackage{bbm}
\usepackage{amsmath,amsthm,amssymb,amsfonts,latexsym}
\usepackage{mathtools}
\usepackage{booktabs}
\usepackage{icomma}
\usepackage{afterpage}
\usepackage{titling}
\thanksmarkseries{alph}
\usepackage[labelfont={small,sc},textfont={small,it},
            margin=15pt,skip=10pt,position=bottom]{caption}
\usepackage[labelfont={scriptsize,normalfont},textfont={scriptsize,it},
            margin=20pt,skip=5pt,position=bottom,
            labelformat=simple]{subcaption}

\usepackage{algorithmic}
\usepackage[boxed]{algorithm}

\newtheorem{remark}{Remark}

\newtheorem{proposition}{Proposition}

\numberwithin{condition}{section}
\numberwithin{assumption}{section}
\numberwithin{remark}{section}
\numberwithin{equation}{section}
\numberwithin{lemma}{section}
\numberwithin{definition}{section}
\numberwithin{theorem}{section}
\numberwithin{proposition}{section}
\numberwithin{table}{section}
\numberwithin{figure}{section}
\numberwithin{theorem}{section}
\numberwithin{corollary}{section}
\numberwithin{property}{section}
\numberwithin{algorithm}{section}

\newcommand{\EQ}{\begin{equation}}
\newcommand{\EN}{\end{equation}}
\newcommand{\EQS}{\begin{equation*}}
\newcommand{\ENS}{\end{equation*}}
\newcommand{\ds}{\displaystyle}

\newcommand*\xbar[1]{%
  \hbox{%
    \vbox{%
      \hrule height 0.5pt
      \kern0.5ex%
      \hbox{%
        \kern-0.1em%
        \ensuremath{#1}%
        \kern-0.1em%
      }%
    }%
  }%
}

\def\n1{n}
\def\argmax{\mathop{\rm arg\,max}}

\newsavebox{\savepar}

\numberwithin{equation}{section}
\numberwithin{table}{section}
\numberwithin{figure}{section}

\def\argmax{\mathop{\rm arg\,max}}

\usepackage[running]{lineno}

\begin{document}
\title{A Stochastic Control Approach to Defined Contribution Plan Decumulation:\\
         {\em ``The Nastiest, Hardest Problem in Finance''} \\
}

\author{Peter A. Forsyth\thanks{David R. Cheriton School of Computer Science,
        University of Waterloo, Waterloo ON, Canada N2L 3G1,
        \texttt{paforsyt@uwaterloo.ca}, +1 519 888 4567 ext.\ 34415.}
  }

\maketitle


\begin{abstract}
We pose the decumulation strategy for a Defined Contribution (DC) pension plan as a problem
in optimal stochastic control.  The controls are the withdrawal amounts and the asset allocation
strategy. We impose maximum and minimum constraints on the withdrawal amounts, and impose no-shorting
no-leverage constraints on the asset allocation strategy.  Our objective function measures reward
as the expected total withdrawals over the decumulation horizon, and risk is measured by 
Expected Shortfall (ES) at the end of the decumulation period.   
We solve the stochastic control problem numerically, based on a parametric
model of market stochastic processes.  We find that, compared to a fixed constant withdrawal strategy,
with minimum withdrawal set to the constant withdrawal amount,
the optimal strategy has a significantly
higher expected average withdrawal, at the cost of a very small increase in ES risk.
Tests on bootstrapped resampled historical market data indicate that this strategy is robust
to parametric model misspecification.

\vspace{5pt}

\end{abstract}

\section{Introduction}
The traditional Defined Benefit (DB) pension plan is in the process of disappearing
for new entrants into the labour market.\footnote{See, 
for example, {\em``The extinction of defined-benefit plans is almost upon us,''}
Globe and Mail, October 4, 2018.  \url{https://www.theglobeandmail.com/investing/personal-finance/retirement/article-the-extinction-of-defined-benefit-pension-plans-is-almost-upon-us/}}
In some countries, notably Australia, DB plans have been replaced by Defined Contribution (DC)  plans
almost exclusively. \footnote{In Australia, DC plans have 86\% of pension assets, compared with
14\% in DB assets.\citep{Towers_2020}}

Assuming the DC plan holder has accumulated a reasonable amount in her DC plan account, the retiree
is faced with an enormous challenge.  The retiree has to devise an investment policy
and a withdrawal strategy during the decumulation phase.
Nobel laureate William Sharpe has referred
to DC plan decumulation 
as {\em``the nastiest, hardest problem in finance''}
\citep{Sharpe2017}.

Although it is often suggested that DC plan holders should purchase annuities upon
retirement, this is rarely done \citep{Peijnenburg2016}.  In fact, \citet{MacDonald2013} argue that
in many instances, this entirely rational.  Reasons for the lack of interest in
annuities include  meager returns of annuities
in the current low interest rate environment, poor annuity pricing, the lack of true inflation protection,
and no access to capital in the event of emergencies.

For an extensive review of strategies for decumulation, we refer the reader to 
\citet{bernhardt-donnelly:2018} and \citet{MacDonald2013}.   A non-exhaustive list of the approaches
discussed by these authors include
use of traditional utility functions, practitioner rules
of thumb, target approaches, minimizing probability of ruin and modern tontines.
Previous decumulation strategies are  also summarized in \citet{forsyth_2020_a}.
Concerning  the current state of DC plan
decumulation strategies, \citet{MacDonald2013} conclude {\em``There is no solution that is appropriate for
everyone and neither is there a single solution for any individual.''}

We should mention that there is a standard rule of thumb for DC plan
decumulation, termed the {\em four per cent rule}.  Based on historical
backtests, \citet{Bengen1994} suggests investing in a portfolio of
50\% bonds and 50\% stocks, and withdrawing 4\% of the initial 
capital each year (adjusted for inflation).  Over historical rolling year
30 year periods, this strategy would have never depleted the
portfolio.

Another recent strategy is based on the Annually Recalculated Virtual Annuity (ARVA)
\citep{Waring2015,West2015,Forsyth_Arva_a}.  The ARVA strategy determines the yearly
spending based on the theoretical value of a fixed term (virtual) annuity purchased
with the current portfolio wealth.  This approach is efficient in the sense
that the portfolio is exhausted at the end of the investment horizon, but
there is no guarantee of a yearly minimum withdrawal amount.

A recent survey\footnote{2017 Allianz Generations Ahead Study 
 - Quick Facts \#1. (2017), Allianz.} showed that
a majority of pre-retirees fear exhausting their savings in retirement more than death.
In addition, it is considered axiomatic amongst practitioners that 
retirees desire to have minimum (real) cash flows each year to fund
expenses \citep{Yamada2011}.  Typical (e.g. CRRA)  utility function based objective functions
do not directly focus on these two issues.

To address these two concerns, 
our objective in this article is to determine a decumulation strategy which
has the following characteristics.
\begin{itemize}
   \item Withdrawals can be variable, but with minimum and maximum constraints.
   \item The risk of portfolio depletion is minimized.
   \item The expected average withdrawal is maximized.
   \item The asset allocation strategy can be dynamic and non-deterministic.
\end{itemize}
We specify that the withdrawals are to take place over a fixed, lengthy (30 years)
decumulation horizon.  We do not explicitly take into account longevity risk,
which we recognize as a weakness of this strategy.  However, this is mitigated
(somewhat) by specifying a long decumulation period.
For example, the probability that a 65-year old Canadian male attains
the age of 95, is about $0.13$. 
\footnote{\url{www.cia-ica.ca/docs/default-source/2014/214013e.pdf}}

We pose the decumulation problem as an exercise in optimal stochastic
control.  We have two controls: the asset allocation and the withdrawal
amount.  These controls are time and state dependent.
Our objective function is composed of a measure of reward
and risk.  Our measure of reward is the expected total withdrawals (EW) over
the thirty year period.  Our measure of risk is expected shortfall (ES)
at the end of the decumulation period.  The ES at level x\% is the mean
of the worst x\% of outcomes.  The negative of ES is also known as
Conditional Value at Risk (CVAR) or Conditional Tail Expectation (CTE).

We emphasize that, in contrast to previous studies on DC plan decumulation \citep{forsyth_2019_a,Forsyth_Arva_a},
where deterministic withdrawal strategies are coupled with optimal asset allocation,
in this work, we determine both the optimal  withdrawal strategy and the optimal asset allocation.

We should note that a strategy using ES as a risk measure
is formally  a pre-commitment policy. 
Some authors have taken the point of view that
pre-commitment polices are not time consistent, hence non-implementable.
However, as noted in \citet{forsyth_2019_c},
the time zero strategy based on a pre-commitment ES policy is
identical to the strategy for an induced time consistent policy,
hence is implementable.\footnote{An implementable strategy has the property that
the investor has no incentive to deviate from the strategy computed at time zero
at later times \citep{forsyth_2019_c}.}  The induced time consistent strategy
in this case is a target based shortfall.  The concept of induced time consistent
strategies is discussed in \citet{Strub_2019_a}. In fact,
\citet{forsyth_2019_c} shows that enforcing a time consistent
constraint on policies which use ES as a risk measure has undesirable
consequences.  The relationship between pre-commitment and implementable
target based schemes in the mean-variance context is discussed in \citet{vigna:2014} and
\citet{Vigna_2017b}.

We assume that the retiree has an investment portfolio consisting
of a stock index and a bond index, and desires to maximize
real (inflation adjusted) total withdrawals.
We calibrate stochastic
models of real stock and bond indexes to historical data over
the 1926:1-2019:12 period.  We assume yearly withdrawals and
rebalancing of the DC account.  We term the market where the
assets follow the parametric model fit to the historical
data the {\em synthetic} market.

We devise a numerical method for determining the optimal policies.
We enforce realistic investment constraints (no shorting, no leverage)
and maximum and minimum constraints on the yearly withdrawal amounts.
Compared to a strategy with a  fixed withdrawal amount per year, we find that a variable withdrawal
strategy, with a minimum withdrawal set to the fixed withdrawal amount,
has a significantly increased expected average withdrawal, with only 
a very small increase in ES risk.

We also test the robustness of the strategy computed in the
synthetic market by carrying out tests using bootstrap
resampled historical data (the {\em historical market}).
The efficient EW-ES frontiers for both synthetic and
historical market tests are very close, indicating
that the strategy computed in the synthetic market is
robust to model misspecification.

\section{Formulation}
We assume that the investor has access to two funds: a broad market stock index fund
and a constant maturity bond index fund.  

The investment horizon is $T$.  Let $S_t$ and $B_t$ respectively denote the 
real (inflation adjusted) \emph{amounts} invested in the
stock index and the bond index respectively.  In general, these
amounts will depend on the investor's strategy over time,
as well as changes
in the real unit prices of the assets.
In the absence of an investor determined
control (i.e. cash withdrawals or rebalancing),
all changes in $S_t$
and $B_t$ result from changes in asset prices. We model the stock index as following
a jump diffusion.  

In addition, we follow the usual practitioner approach and directly model
the returns of the constant maturity bond index as a stochastic process,
see for example \citet{Lin_2015,mitchell_2014}.   
As in \citet{mitchell_2014}, we assume that the constant maturity bond
index follows a jump diffusion process as well.

Let $ S_{t^-} =   S(t - \epsilon), 
\epsilon \rightarrow 0^+$, i.e.\ $t^-$ is the instant of time before
$t$, and let $\xi^s$ be a random number representing a jump multiplier.
When a jump occurs, $S_t = \xi^s S_{t^-}$. 
Allowing for jumps permits modelling of non-normal asset returns.
We assume that $\log(\xi^s)$ follows a double exponential distribution
\citep{kou:2002,Kou2004}. If a jump occurs, $p_{\text{\emph{u}}}^s$ is
the probability of an upward jump, while $1-p_{\text{\emph{u}}}^s$ is the
chance of a downward jump. The density function for $y = \log (\xi^s)$ is
\begin{equation}
f^s(y) = p_{u}^s \eta_1^s e^{-\eta_1^s y} {\bf{1}}_{y \geq 0} +
       (1-p_{u}^s) \eta_2^s e^{\eta_2^s y} {\bf{1}}_{y < 0}~.
\label{eq:dist_stock}
\end{equation}
We also define
\begin{eqnarray}
\kappa^s &= &E[ \xi^s -1 ]
          =   \frac{p_{\text{\emph{{u}}}}^s \eta_1^s}{\eta_1^s - 1} + 
          \frac{ ( 1 - p_{ \text{\emph{u}} }^s ) \eta_2^s }{\eta_2^s + 1}  -1 ~.
\end{eqnarray}
In the absence of control, $S_t$ evolves according to
\begin{eqnarray}
\frac{dS_t}{S_{t^-}} &= &\left(\mu^s -\lambda_\xi^s \kappa_{\xi}^s \right) \, dt + 
  \sigma^s \, d Z^s +  d\left( \ds \sum_{i=1}^{\pi_t^s} (\xi_i^s -1) \right) ,
\label{jump_process_stock}
\end{eqnarray}
where $\mu^s$ is the (uncompensated) drift rate, $\sigma^s$ is the volatility,
$d Z^s$ is the increment of a Wiener process,
$\pi_t^s$ is a Poisson process with positive intensity parameter
$\lambda_\xi^s$, and $\xi_i^s$ are i.i.d.\ positive random variables having
distribution (\ref{eq:dist_stock}).
Moreover, $\xi_i^s$, $\pi_t^s$, and $Z^s$ are assumed to all be
mutually independent.

Similarly,  let the amount in the bond index be $B_{t^-} =  B(t - \epsilon), \epsilon \rightarrow 0^+$.
In the absence of control, $B_t$ evolves as
\begin{eqnarray}
\frac{dB_t}{B_{t^-}} &= &\left(\mu^b -\lambda_\xi^b \kappa_{\xi}^b  
   + \mu_c^b {\bf{1}}_{\{B_{t^-} < 0\}}  \right) \, dt + 
  \sigma^b \, d Z^b +  d\left( \ds \sum_{i=1}^{\pi_t^b} (\xi_i^b -1) \right) ,
\label{jump_process_bond}
\end{eqnarray}
where the terms in equation (\ref{jump_process_bond}) are defined analogously to
equation (\ref{jump_process_stock}).  In particular, $\pi_t^b$ 
is a Poisson process with positive intensity parameter
$\lambda_\xi^b$, and $\xi_i^b$ has distribution 
\begin{equation}
f^b( y= \log \xi^b) = p_{u}^b \eta_1^b e^{-\eta_1^b y} {\bf{1}}_{y \geq 0} +
       (1-p_{u}^b) \eta_2^b e^{\eta_2^b y} {\bf{1}}_{y < 0}~,
\label{eq:dist_bond}
\end{equation}
and $\kappa_{\xi}^b = E[ \xi^b -1 ]$.  $\xi_i^b$, $\pi_t^b$, and $Z^b$ are assumed to all be
mutually independent.  The term $\mu_c^b {\bf{1}}_{\{B_{t^-} < 0\}}$ in equation 
(\ref{jump_process_bond}) represents the extra cost of borrowing (the spread).

The diffusion processes are correlated, i.e. $d Z^s \cdot d Z^b = \rho_{sb} dt$.  The stock
and bond jump processes are assumed mutually independent.
See \citet{forsyth_2020_a} for justification of the assumption of stock-bond jump independence.

\begin{remark}[Stock and Bond Processes]
An obvious generalization of processes
(\ref{jump_process_stock}) and (\ref{jump_process_bond}) would be
to include stochastic volatility effects.  However,
previous studies have shown that stochastic volatility appears to have little consequences for 
long term investors \citep{Ma2015}.  As a robustness check, we will 
(i) determine the optimal controls using the parametric model based
on equations (\ref{jump_process_stock}) and (\ref{jump_process_bond}) and (ii) use
these controls on bootstrapped resampled historical data, which makes no assumptions
about the underlying bond and stock stochastic processes.
\end{remark}

We define the investor's total wealth at time $t$ as
\begin{equation}
\text{Total wealth } \equiv W_t = S_t + B_t.
\end{equation}
We impose the constraints that (assuming solvency) shorting stock and using leverage
(i.e.\ borrowing) are not permitted, which would be typical of a
DC plan retirement savings account.  In the event of insolvency (due to withdrawals), the portfolio
is liquidated, trading ceases and debt accumulates
at the borrowing rate.

\section{Notational conventions}
\label{adaptive_section}
Consider a set of discrete withdrawal/rebalancing times $\mathcal{T}$
\begin{eqnarray}
   \mathcal{T} = \{t_0=0 <t_1 <t_2< \ldots <t_M=T\}  \label{T_def}
\end{eqnarray}
where we assume that $t_i - t_{i-1} = \Delta t =T/M$ is constant for
simplicity.
To avoid subscript clutter, in the following, we will
occasionally use the notation $S_t \equiv S(t), B_t \equiv B(t)$ and
$W_t \equiv W(t)$.
Let the inception time of the investment be $t_0 = 0$. We let
$\mathcal{T} $ be the set of
withdrawal/rebalancing times, as defined in equation (\ref{T_def}).
At each rebalancing time
$t_i$, $i = 0, 1, \ldots, M-1$, the investor (i)~withdraws an amount of cash
$q_i$ from the portfolio, and then~(ii) rebalances the portfolio. At
$t_M = T$,  the final cash flow $q_M$ occurs, and the portfolio is liquidated.
In the following, given a time dependent function $f(t)$, then we will use
the shorthand notation
\begin{eqnarray}
  f(t_i^+) \equiv \displaystyle  \lim_{\epsilon \rightarrow 0^+}
          f(t_i + \epsilon) ~~&; & ~~
      f(t_i^-) \equiv \displaystyle  \lim_{\epsilon \rightarrow 0^+}
          f(t_i - \epsilon)  ~~.
\end{eqnarray}
We assume that there are no
taxes or other transaction costs, so that the condition
\begin{eqnarray}
   W(t_i^+) = W(t_i^-) - q_i &;& t_i \in \mathcal{T}
\end{eqnarray}
holds.  Typically, DC plan savings are held in a tax advantaged
account, with no taxes triggered by rebalancing.  With infrequent (e.g. yearly) rebalancing, we also
expect transaction costs to be small, and hence can be ignored.  It is possible to include
transaction costs, but at the expense of increased computational cost \citep{Van2018}.

We denote by
$X\left(t\right)=
\left( S \left( t \right), B\left( t \right) \right)$,
$t\in\left[0,T\right]$, the multi-dimensional controlled
underlying process, and by $x = (s, b)$
the realized state of the system.
Let the rebalancing control $p_i(\cdot)$ be  the fraction invested in the stock index
at the rebalancing date $t_i$, i.e.
\begin{eqnarray}
    p_i \left( X(t_i^-) \right) = 
       p \left( X(t_i^-), t_i \right) & = & \frac{ S(t_i^+)} {S(t_i^+) + B(t_i^+) } ~.
\end{eqnarray}

Let the withdrawal control $q_i(\cdot)$ be the amount withdrawn at time $t_i$, i.e.
$q_i \left( X(t_i^-) \right)  = q \left( X(t_i^-), t_i \right)$.
Note that formally,  the controls depend on the state
of the investment portfolio, before the rebalancing
occurs, i.e.
$p_i(\cdot) =  p\left(X(t_i^-), t_i)\right) 
= p\left(X_i^-, t_i \right)$, and
$q_i(\cdot) =  q\left(X(t_i^-), t_i)\right) 
= q\left(X_i^-, t_i \right)$,
$t_i \in \mathcal{T}$, where $\mathcal{T}$ is the
set of rebalancing times.

However, it will be convenient to note that
in our case, we find the optimal control $p_i(\cdot)$
amongst all strategies with constant wealth (after withdrawal of cash).
Hence, with some abuse of notation, we will now consider
$p_i(\cdot)$ to be function of wealth after withdrawal of cash
\begin{eqnarray}
  p_i(\cdot) &= & p(W(t_i^+), t_i) \nonumber \\
      & & W(t_i^+) = S(t_i^-) + B(t_i^-) - q_i(\cdot) \nonumber \\
      & &  S(t_i^+) = S_i^+ = p_i(W_i^+)~ W_i^+ ~~;~~ B(t_i^+) = B_i^+ =   (1 -p_i(W_i^+)) ~W_i^+
       ~~.  \label{p_def_2}
\end{eqnarray}
\begin{remark}[Control depends on wealth only]
Note that we assume no transaction costs.  If transaction costs are included, then the control $p_i(\cdot)$ would
in general be a function of the state $(S(t_i^-), B(t_i^-))$ \citep{dang-forsyth:2014a}.
\end{remark}
Note that since $p_i(\cdot) = p_i(W_i^- -q_i)$, then it follows that
\begin{eqnarray}
  q_i(\cdot) = q_i(W_i^-) \label{q_dependence}
\end{eqnarray}
which we will prove formally in a later section.

\begin{remark}[Instantaneous Rebalancing]
We assume that rebalancing occurs instantaneously.  Informally, this
has the consequence  that no jumps occur in the
unit prices of the stock and bond indexes over a  rebalancing period $(t_i^-, t_i^+)$.
\end{remark}

A control at time $t_i$, is then given by the pair $( q_i(\cdot), p_i(\cdot) )$ where the notation $(\cdot)$
denotes that the control is a function of the state.

Let $\mathcal{Z}$ represent the set of admissible
values of the controls $(q_i(\cdot), p_i(\cdot))$.
As is typical for a DC plan savings account, we impose no-shorting, no-leverage
constraints (assuming solvency).  We also impose maximum and minimum values for the withdrawals.
We apply the constraint that in the event of insolvency due to withdrawals ($W(t_i^+) < 0$),
trading ceases and debt (negative wealth) accumulates at the appropriate
bond rate of return (including a spread).  
We also specify that the stock assets are liquidated at $t=t_M$.

More precisely, let $W_i^+$ be the wealth after withdrawal of cash, then
define
\begin{eqnarray}
 \mathcal{Z}_q  & = &
                  [q_{\min}, q_{\max} ] ~;~   t \in \mathcal{T} 
             ~,  \label{Z_q_def}\\
    \mathcal{Z}_p (W_i^+,t_i) &=&
          \begin{cases}
                  [0,1] & W_i^+ > 0 ~;~ t_i \in \mathcal{T}~;~ t_i \neq t_M \\
                  \{0\} & W_i^+ \leq 0 ~;~ t_i \in \mathcal{T}~;~  t_i \neq t_M \\
                  \{0\} &  t_i=t_M
          \end{cases}    ~.  \label{Z_p_def} \\
\end{eqnarray}

The set of admissible values for $(q_i,p_i), t_i \in \mathcal{T}$,
can then be written a
\begin{eqnarray}
   (q_i,p_i) \in \mathcal{Z}(W_i^+,t_i) & = & \mathcal{Z}_q \times \mathcal{Z}_p (W_i^+,t_i)~.
  \label{admiss_set}
\end{eqnarray}
For implementation purposes, we have written equation (\ref{admiss_set}) in terms of the wealth after
withdrawal of cash.  However, we remind the reader that since $W_i^+ = W_i^- -q$, the controls
are formally a function of the state $X(t_i^-)$ before the control is applied.

The admissible control set $\mathcal{A}$ can then be written as
\begin{eqnarray}
  \mathcal{A} = \biggl\{
                  (q_i, p_i)_{0 \leq i \leq M} : (p_i, q_i) \in \mathcal{Z}(W_i^+,t_i) 
                \biggr\}
\end{eqnarray}
An admissible control $\mathcal{P} \in \mathcal{A}$, where $\mathcal{A}$ is
the admissible control set, can be written as,
\begin{eqnarray}
    \mathcal{P} = \{ (q_i(\cdot), p_i(\cdot) ) ~:~ i=0, \ldots, M \} ~.
\end{eqnarray}
We also define $\mathcal{P}_n \equiv \mathcal{P}_{t_n} \subset \mathcal{P}$
as the tail of the set of controls in $[t_n, t_{n+1}, \ldots, t_{M}]$, i.e.
\begin{eqnarray}
   \mathcal{P}_n =\{ (   q_n(\cdot) , p_n(\cdot) )  \ldots, (p_{M}(\cdot), q_{M}(\cdot) ) \} ~.
\end{eqnarray}
For notational completeness, we also define the tail of the admissible control set $\mathcal{A}_n$ as
\begin{eqnarray}
   \mathcal{A}_n = \biggl\{
                  ( q_i, p_i )_{n \leq i \leq M} : (q_i, p_i) \in \mathcal{Z}(W_i^+,t_i) 
                \biggr\}
\end{eqnarray}
so that $\mathcal{P}_n \in \mathcal{A}_n $.

\section{Risk and reward}
\subsection{A measure of risk: definition of expected shortfall (ES)}
Let $g(W_T)$ be the probability density function of wealth $W_T$ at $t=T$.
Suppose
\begin{linenomath*}
\begin{equation}
\int_{-\infty}^{W^*_{\alpha}}  g(W_T) ~dW_T = \alpha,
    \label{CVAR_def_a}
\end{equation}
\end{linenomath*}
i.e.\ \emph{Pr}$[W_T >  W^*_{\alpha}] = 1 - \alpha$.  We can interpret
$W^*_{\alpha}$ as the Value at Risk (VAR) at level $\alpha$ \footnote{In practice, the negative of $W^*_{\alpha}$ is often the
reported VAR.}.
The Expected Shortfall (ES) at level $\alpha$ is then
\begin{linenomath*}
\begin{equation}
{\text{ES}}_{\alpha} = \frac{\int_{-\infty}^{W^*_{\alpha}} W_T ~  g(W_T) ~dW_T }
                      {\alpha},
\label{ES_def_1}
\end{equation}
\end{linenomath*}
which is the mean of the worst $\alpha$ fraction of outcomes.
Typically $\alpha \in \{ .01, .05 \}$. Note that the definition of ES in
equation \eqref{ES_def_1} uses the probability density of the final
wealth distribution, not the density of \emph{loss}. Hence, in our case,
a larger value of ES (i.e.\ a larger value of average worst case terminal
wealth) is desired.  The negative of ES is commonly referred to as Conditional
Value at Risk (CVAR).

Define $X_0^+ = X(t_0^+), X_0^- = X(t_0^-)$.  Given an expectation under control $\mathcal{P}$, $E_{\mathcal{P}}
[\cdot]$, as noted by \citet{Uryasev_2000},
$ \text{ES}_{\alpha}$ can be alternatively written as
\begin{eqnarray}
  {\text{ES}}_{\alpha}( X_0^-, t_0^-)  & = & 
      \sup_{W^*} E_{\mathcal{P}_0}^{X_0^+, t_0^+}
   \biggl[W^* + \frac{1}{\alpha} \min(  W_T - W^* , 0 )
                       \biggr]  ~.
\label{ES_def}
\end{eqnarray}
The admissible set for $W^*$ in equation (\ref{ES_def}) is over
the set of possible values for $W_T$.  

Note that the notation ${\text{ES}}_{\alpha}( X_0^-, t_0^-) $
emphasizes that  ${\text{ES}}_{\alpha}$ is as seen at  $(X_0^-, t_0^-)$.  In other words,
this is the pre-commitment ${\text{ES}}_{\alpha}$.  A strategy based purely on optimizing
the pre-commitment value of ${\text{ES}}_{\alpha}$ at time zero  is {\em time-inconsistent},
hence has been termed by many as {\em non-implementable}, since the investor
has an incentive to deviate from the the pre-commitment strategy at $t >0$.
However, in the following, we will consider the pre-commitment strategy merely
as a device to determine an appropriate level of $W^*$ in equation (\ref{ES_def}).
If we fix $W^*$ $\forall t >0$, then this strategy is the induced time consistent
strategy \citep{Strub_2019_a}, hence is implementable.  We delay further discussion of this subtle point
to later sections.

\subsection{A measure of reward: expected total withdrawals (EW)}
We will use expected total withdrawals as a measure of reward in the following. More precisely, we define EW (expected withdrawals) as
\begin{eqnarray}
      {\text{EW}}( X_0^-, t_0^-) = E_{\mathcal{P}_0}^{X_0^+, t_0^+} 
                                 \biggl[ 
                                    \sum_{i=0}^{i=M} q_i
                                 \biggr] ~.  \label{EW_def}
\end{eqnarray}

\section{Objective function}
Since expected withdrawals (EW) and expected shortfall (ES) are conflicting measures,
we use a scalarization technique to find the Pareto points for this multi-objective
optimization problem.  Informally, for a given scalarization parameter $\kappa >0$,
we seek to find the control $\mathcal{P}_0$ that maximizes
\begin{eqnarray}
 {\text{EW}}( X_0^-, t_0^-)  + \kappa {\text{ES}}_{\alpha}( X_0^-, t_0^-)
     ~.  \label{objective_overview}
\end{eqnarray}
More precisely, we define the pre-commitment EW-ES problem $(PCEE_{t_0}(\kappa))$
problem in terms of the value function $J(s,b,t_0^-)$ 

\begin{eqnarray}
\left(\mathit{PCEE_{t_0}}\left(\kappa \right)\right):
    \qquad J\left(s,b,  t_{0}^-\right)  
    & = & \sup_{\mathcal{P}_{0}\in\mathcal{A}}
          \sup_{W^*}
        \Biggl\{
               E_{\mathcal{P}_{0}}^{X_0^+,t_{0}^+}
           \Biggl[ ~\sum_{i=0}^{i=M} q_i ~  + ~
              \kappa \biggl( W^* + \frac{1}{\alpha} \min (W_T -W^*, 0) \biggr)
                    \Biggr. \Biggr. \nonumber \\
         & &  \Biggl. \Biggl. ~~~~~
                \bigg\vert X(t_0^-) = (s,b)
                   ~\Biggr] \Biggr\}\label{PCEE_a}\\
    \text{ subject to } & &
               \begin{cases}
(S_t, B_t) \text{ follow processes \eqref{jump_process_stock} and \eqref{jump_process_bond}};  
 ~~t \notin \mathcal{T} \\
      W_{\ell}^+ = S_{\ell}^{-} + B_{\ell}^{-}  - q_\ell \,; ~ X_\ell^+ = (S_\ell^+ , B_\ell^+)  \\
   S_\ell^+ = p_\ell(\cdot) W_\ell^+ \,; 
 ~B_\ell^+ = (1 - p_\ell(\cdot) ) W_\ell^+ \, \\
    ( q_\ell(\cdot) , p_\ell(\cdot) )  \in \mathcal{Z}(W_\ell^+,t_\ell)  \\
    \ell = 0, \ldots, M ~;~ t_\ell \in \mathcal{T}  \\
               \end{cases}  ~.
\label{PCEE_b}
\end{eqnarray}

Interchange the $\sup \sup$ in equation (\ref{PCEE_a}), so that
value function $ J\left(s,b,  t_{0}^-\right)$ can be
written as
\begin{eqnarray}
\qquad J\left(s,b,  t_{0}^-\right)
     & = &  \sup_{W^*} \sup_{\mathcal{P}_{0}\in\mathcal{A}}
              \Biggl\{
               E_{\mathcal{P}_{0}}^{X_0^+,t_{0}^+}
           \Biggl[ ~\sum_{i=0}^{i=M} q_i ~  + ~
              \kappa \biggl( W^* + \frac{1}{\alpha} \min (W_T -W^*, 0) \biggr)
                \bigg\vert X(t_0^-) = (s,b)
                   ~\Biggr] \Biggr\} ~.  \nonumber \\
             \label{pcee_inter}
\end{eqnarray}
Noting that the inner supremum in equation (\ref{pcee_inter}) is a continuous
function of $W^*$, and noting that the optimal value of $W^*$ in equation (\ref{pcee_inter})
is bounded\footnote{This is the same as noting that a finite value at risk exists.  This easily
shown,  assuming $ 0 <\alpha < 1$, since our investment strategy uses no leverage and no-shorting.}, then define
\begin{eqnarray}
\mathcal{W}^*(s,b)  & = & \displaystyle \argmax_{W^*} \biggl\{
                      \sup_{\mathcal{P}_{0}\in\mathcal{A}} 
            \Biggl\{
               E_{\mathcal{P}_{0}}^{X_0^+,t_{0}^+}
           \Biggl[ ~\sum_{i=0}^{i=M} q_i ~  + ~
              \kappa \biggl( W^* + \frac{1}{\alpha} \min (W_T -W^*, 0) \biggr)
                \bigg\vert X(t_0^-) = (s,b)
                   ~\Biggr] \Biggr\} ~.
             \nonumber \\
            \label{pcee_argmax}
\end{eqnarray}
We refer the reader to \citet{forsyth_2019_c} for an extensive discussion
concerning pre-commitment and time consistent ES strategies.  We summarize
the relevant results from that research here.

Denote the investor's initial wealth at $t_0$ by $W_0^-$.  Then we have the following
result.
\begin{proposition}[Pre-commitment strategy equivalence to a time consistent
policy for an alternative objective function]\label{equiv_thm}
The pre-commitment EW-ES strategy $\mathcal{P}^*$ determined by solving
$J(0, W_0, t_0^-)$ (with $\mathcal{W}^*( 0, W_0^-)$ from equation (\ref{pcee_argmax}))
is the time consistent strategy for the equivalent
problem $TCEQ$ (with fixed $\mathcal{W}^*(0,W_0^-)$), with value function $\tilde{J}(s,b,t)$
defined by 
\begin{eqnarray}
 \left(\mathit{TCEQ_{t_n}}\left(\kappa / \alpha  \right)\right):
    \qquad \tilde{J}\left(s,b,  t_{n}^-\right)  
    & = &
            \sup_{\mathcal{P}_{n}\in\mathcal{A}}
        \Biggl\{
               E_{\mathcal{P}_{n}}^{X_n^+,t_{n}^+}
           \Biggl[ ~\sum_{i=n}^{i=M} q_i ~  + ~
               \frac{\kappa}{\alpha} \min (W_T -\mathcal{W}^*(0, W_0^-),0) 
                    \Biggr. \Biggr. \nonumber \\
         & &  \Biggl. \Biggl. ~~~~~ 
                \bigg\vert X(t_n^-) = (s,b)
                   ~\Biggr] \Biggr\}~.
             \label{timec_equiv}
\end{eqnarray}

\end{proposition}
\begin{proof}
This follows  similar steps as in \citet{forsyth_2019_c}, proof of Proposition 6.2, with the exception
that the reward in \citet{forsyth_2019_c} is expected terminal wealth, while here the reward is
total withdrawals.
\end{proof}

\begin{remark}[An Implementable Strategy]
Given an initial level of wealth $W_0^-$ at $t_0$, then the optimal control
for the pre-commitment problem (\ref{PCEE_a}) is the same optimal control
for the time consistent problem $\left(\mathit{TCEQ_{t_n}}\left(\kappa / \alpha  \right)\right)$
(\ref{timec_equiv}), $\forall t >0$.  Hence we can regard problem
$\left(\mathit{TCEQ_{t_n}}\left(\kappa / \alpha  \right)\right)$
as the {\em EW-ES induced time consistent strategy}.
Thus, the induced strategy is implementable,
in the sense that the investor has no incentive to deviate
from the strategy computed at time zero, at later times \citep{forsyth_2019_c}.
\end{remark}

\begin{remark}[EW-ES Induced Time Consistent Strategy]
In the following, we will consider the  actual strategy followed by
the investor for any $t>0$ as given by the induced time consistent
strategy $\left(\mathit{TCEQ_{t_n}}\left(\kappa / \alpha  \right)\right)$
in equation (\ref{timec_equiv}), with a fixed value of $\mathcal{W}^*(0, W_0^-)$,
which is identical to the EW-ES strategy at time zero.
Hence, we will refer to this strategy in the following as the EW-ES strategy, 
with the understanding that this refers to strategy 
$\left(\mathit{TCEQ_{t_n}}\left(\kappa / \alpha  \right)\right)$
for any $t>0$.
\end{remark}

\section{Algorithm for optimal expected-withdrawals expected-shortfall (EW-ES) strategy}
\label{algo_section}

\subsection{Formulation}\label{formulation_section}
In order to solve problem $(PCEE_{t_0}(\kappa))$, our
starting point is equation (\ref{pcee_inter}),
where we have interchanged the $\sup \sup (\cdot) $ in
equation (\ref{PCEE_a}).  We expand the state space to $\hat{X} = (s,b,W^*)$,
and define the auxiliary function $V(s, b, W^*, t) \in \Omega = [0,\infty) \times (-\infty, +\infty) \times (-\infty, +\infty)
\times [0, \infty)$
\begin{eqnarray}
   V(s, b, W^*, t_n^-) & = & \sup_{\mathcal{P}_{n}\in\mathcal{A}_n}
        \Biggl\{
               E_{\mathcal{P}_{n}}^{\hat{X}_n^+,t_{n}^+}
           \Biggl[
                \sum_{i=n}^{i=M} q_i +
           \kappa
             \biggl(
                  W^* + \frac{1}{\alpha} \min((W_T -W^*),0) 
               \biggr)
              \bigg\vert \hat{X}(t_n^-) = (s,b, W^*)  \Biggr]
                   \Biggr\}~. \nonumber \\
                    ~ \label{expanded_1} \\
     \text{ subject to } & &
  \begin{cases}
(S_t, B_t) \text{ follow processes \eqref{jump_process_stock} and \eqref{jump_process_bond}};  
 ~~t \notin \mathcal{T} \\
      W_{\ell}^+ = S_{\ell}^{-} + B_{\ell}^{-}  - q_\ell \,; ~ X_\ell^+ = (S_\ell^+ , B_\ell^+)  \\
   S_\ell^+ = p_\ell(\cdot) W_\ell^+ \,; 
 ~B_\ell^+ = (1 - p_\ell(\cdot) ) W_\ell^+ \, \\
    ( q_\ell(\cdot), p_\ell(\cdot) )  \in \mathcal{Z}(W_\ell^+ ,t_\ell)   \\
    \ell = n, \ldots, M ~;~ t_\ell \in \mathcal{T}  \\
               \end{cases}  ~.
             \label{expanded_2}
\end{eqnarray}

Equation (\ref{expanded_1}) is a simple expectation.  Hence we can solve this auxiliary
problem using dynamic programming.  
Recalling the definitions of $\mathcal{Z}_p, \mathcal{Z}_q$ in
equations (\ref{Z_q_def}-\ref{Z_p_def}), then
the dynamic programming principle applied at $t_n \in \mathcal{T}$ would then imply
\begin{eqnarray}
     V(s,b,W^*, t_n^-) & = & 
                           \sup_{q \in \mathcal{Z}_q}  ~~~\sup_{p \in \mathcal{Z}_p(w^- - q, t) } 
                         \biggl\{ q+
                              \biggl[ 
                                    V(  (w^- -q) p ,  (w^- - q) (1-p) , W^*, t_n^+)
                                     \biggr]
                                                 \biggr\} \nonumber \\
                       & = & \sup_{ q \in Z_q}  \biggl\{ q + 
                                   \biggl[ \sup_{ p \in \mathcal{Z}_p( w^- - q, t)}
                                    V(  (w^- -q) p ,  (w^- - q) (1-p) , W^*, t_n^+)
                                     \biggr]
                                                 \biggr\} \nonumber \\
                        & & w^- = s+b  ~ . \label{dynamic_a}
\end{eqnarray}

Let $\xbar{V}$ denote the upper semi-continuous envelope of $V$.
The optimal control $p_n(w,W^*)$ at time $t_n$
is then determined from 
\begin{eqnarray}
     p_n(w, W^*) & = & \left\{
            \begin{array}{lc}
                      \displaystyle  
            \argmax_{ p^{\prime} \in [0,1] } {\xbar{V}}( w p^{\prime}, w (1 - p^{\prime}), W^*, t_n^+),
                          &  w > 0 ~;~ t_n \neq t_M\\
                           0, &  w   \leq 0 ~{\mbox{ or }} t_n = t_M
             \end{array}
                  \right.  ~.
             \label{expanded_3}
\end{eqnarray}
The control for $q$ is then determined from
\begin{eqnarray}
      q_n(w,W^*) & = & \argmax_{ q^{\prime} \in \mathcal{Z}_q }
                    \biggl\{
                       q^{\prime} 
                     + {\xbar{V}}( (w -q^{\prime}) p_n( w -q^{\prime} ) , W^*) ,
                          (w -q^{\prime}) ( 1 - p_n( w -q^{\prime} ) ), W^*) ,t_n^+)
                     \biggr\}  ~. \nonumber \\
                  \label{q_opt}
\end{eqnarray}
From the right hand sides of equation (\ref{expanded_3}) and equation (\ref{q_opt}), we have the following
result.
\begin{proposition}[Dependence of optimal controls]
For fixed $W^*$, the optimal control for $q_n(\cdot)$ is a function  only of the total
portfolio wealth before withdrawals $w^- = s + b$, i.e. $q_n = q_n(w^-, W^*)$, while the optimal control for $p_n(\cdot)$ is a
function only of the total portfolio wealth after withdrawals $w^+ = w^- - q_n( w^-, W^*)$, i.e. $p_n = p_n(w^+, W^*)$.
\end{proposition}
The solution is advanced (backwards) across time $t_n$ by
\begin{eqnarray}
    V( s, b, W^*,t_n^-) & = & q_n( w^-, W^*) + {\xbar{V}}( w^+ p_n(w^+,W^*), w^+(~ 1 - p_n(w^+,W^*) ~), W^*, t_n^+)
                           \nonumber \\
                        & &~~~~~~ w^- = s + b~;~ w^+ = s + b - q_n(w^-, W^*) ~. \nonumber \\
                 \label{expanded_4}
\end{eqnarray}
At $t=T$, we have
\begin{eqnarray}
V(s, b, W^*,T^+) & = &  \kappa \biggl( 
                               W^* + \frac{\min( (s+b -W^*), 0 )}{\alpha} 
                               \biggr) ~. \label{expanded_5}
\end{eqnarray}
For $t \in(t_{n-1},t_n)$, there are no cash flows, discounting (all quantities are inflation
adjusted), or controls applied.  Hence the tower property gives
for $0 <h < (t_n - t_{n-1})$
\begin{eqnarray}
   V(s,b,W^*, t) & = & E\biggl[ 
                   V( S(t+h), B(t+h), W^*, t+h) \big\vert S(t) = s, B(t) = b 
                         \biggr] ~;~ t \in( t_{n-1}, t_n -h) ~. \nonumber \\
        \label{expanded_6}
\end{eqnarray}
Applying Ito's Lemma for jump processes \citep{Cont_book}, noting equations (\ref{jump_process_stock})
and (\ref{jump_process_bond}), and letting $h \rightarrow 0$  gives,
for $t \in (t_{n-1}, t_n)$
\begin{eqnarray}
& & V_t +  \frac{ (\sigma^s)^2 s^2}{2} V_{ss} +( \mu^s - \lambda_{\xi}^s \kappa_{\xi}^s) s V_s
       + \lambda_{\xi}^s \int_{-\infty}^{+\infty} V( e^ys, b, t) f^s(y)~dy 
      + \frac{ (\sigma^b)^2 b^2}{2} V_{bb} 
                  \nonumber \\
     & & ~~~     + ( \mu^b - \lambda_{\xi}^b \kappa_{\xi}^b) b V_b       
                  + \lambda_{\xi}^b \int_{-\infty}^{+\infty} V( s, e^yb, t) f^b(y)~dy 
       -( \lambda_{\xi}^s + \lambda_{\xi}^b )  V + \rho_{sb} \sigma^s \sigma^b s b V_{sb} 
       =0  ~. \nonumber \\
      \label{expanded_7}
\end{eqnarray}

\begin{proposition}[Equivalence of formulation (\ref{expanded_1}-\ref{expanded_7}) to problem $(PCEE_{t_0}(\kappa))$]
Define 
\begin{eqnarray}
  J(s,b,t_0^-) = \sup_{W^\prime} V(s,b,W^{\prime},t_0^-)~,
             \label{expanded_equiv}
\end{eqnarray}
then formulation (\ref{expanded_1}-\ref{expanded_7}) is equivalent to problem $(PCEE_{t_0}(\kappa))$.
\end{proposition}
\begin{proof}
Replace $V(s,b,W^{\prime},t_0^-)$ in equation (\ref{expanded_equiv}) by the expressions in
equations (\ref{expanded_1}-\ref{expanded_7}).  Begin with  equation (\ref{expanded_5}), and recursively
work backwards in time, then we obtain equations (\ref{PCEE_a}-\ref{PCEE_b}), 
by interchanging  $\sup_{W^{\prime}} \sup_{\mathcal{P}} $ in the final step.
\end{proof}

\section{Continuous withdrawal/rebalancing limit}
In order to develop some intuition about the nature of the optimal controls,
we will examine the limit as the rebalancing interval becomes vanishingly small.

\begin{proposition}[Bang-bang withdrawal control in the continuous withdrawal limit]
\label{bang_bang_prop}
Assume that
\begin{itemize}
   \item the stock and bond processes follow (\ref{jump_process_stock})
and (\ref{jump_process_bond}),

  \item the portfolio is continuously rebalanced, and withdrawals occur at
     a continuous (finite) rate $\hat{q} \in [\hat{q}_{\min}, \hat{q}_{\max}]$,

  \item the HJB equation for the EW-ES problem in the continuous rebalancing limit has
        bounded derivatives w.r.t. total wealth,

  \item in the event of ties for the control $\hat{q}$, the smallest withdrawal is selected,
\end{itemize}
then the optimal withdrawal control $\hat{q}^*(\cdot)$ for the EW-ES problem
$(PCEE_{t_0}(\kappa))$ is bang-bang,
$\hat{q}^* \in \{\hat{q}_{\min}, \hat{q}_{\max} \}$.
\end{proposition}

\begin{proof}
We consider (for ease of exposition) the case where the stock and bond funds follow geometric
Brownian motion (i.e. no jumps).  The analysis below can be easily (although tediously)
extended to the case of processes (\ref{jump_process_stock}) and (\ref{jump_process_bond}).
Consequently, we assume that the stock $S_t$ and bond $B_t$ index processes are
\begin{eqnarray}
\frac{dS_t}{S_{t}} = \mu^s~  dt +
  \sigma^s ~ d Z^s ~ & ; &
\frac{dB_t}{B_{t}} = \mu^b ~dt
  + \sigma^b ~ d Z^b  ~.
\label{gbm_process} 
\end{eqnarray}
with $d Z^s \cdot d Z^b = \rho_{sb}~dt$.
Assume that rebalancing is carried out continuously, and let
\begin{eqnarray}
   \hat{p}( W_t, t) & = & \frac{ S_t}{S_t + B_t}~,
\end{eqnarray}
with continuous withdrawal of cash at a rate of $\hat{q}(W_t, t)$.   The SDE for the total wealth process
$W_t = S_t + B_t$ is then
\begin{eqnarray}
   dW_t & = &  \hat{p} W_t  \frac{dS_t}{S_{t}} + (1-\hat{p})  W_t  \frac{dB_t}{B_{t}} -\hat{q}~dt ~.
\end{eqnarray}
It is important to note that here $\hat{q}$ is a {\em rate} of cash withdrawal, whereas
we have previously defined $q$ as a finite {\em amount} of cash withdrawal.
Define the following  sets of admissible values
of the controls
\begin{eqnarray}
   \hat{\mathcal{Z}}_q  & = &
                  [ \hat{q}_{\min}, \hat{q}_{\max} ] ~;~   t \in [0,T] 
             ~,  \label{Z_q_def_cont} \\
    \hat{\mathcal{Z}}_p (W_t,t) &=&
          \begin{cases}
                  [0,1] & W_t > 0 ~;~ t \in [0,T]~;~ t \neq T \\
                  \{0\} & W_t \leq 0 ~;~ t \in [0,T]~;~  t \neq T \\
                  \{0\} &  t=T
          \end{cases}    ~.  \label{Z_p_def_cont} 
\end{eqnarray}
We define the value function $\hat{V}(w, W^*, t)$ on the domain $\hat{\Omega} = 
(-\infty, +\infty) \times (-\infty, +\infty) \times (0, \infty) $ for fixed $W^*$ as
\begin{eqnarray}
& & \hat{V}(w, W^*, t)  =  \nonumber \\
& & \sup_{ \hat{p}(\cdot) \in \hat{\mathcal{Z}}_p}
                       \sup_{ \hat{q}(\cdot) \in \hat{\mathcal{Z}}_q}
        \Biggl\{
               E_{( \hat{p}, \hat{q}) }^{(W_t, W^*,t)}
           \Biggl[
                \int_t^T \hat{q}~dt +
           \kappa
             \biggl(
                  W^* + \frac{1}{\alpha} \min((W_T -W^*),0)
               \biggr)
              \bigg\vert (W_t,W^*) = (w, W^*)  \Biggr]
                   \Biggr\}~. \nonumber \\
                    ~ \label{expanded_1_cont} 
\end{eqnarray}
The continuous rebalancing, continuous withdrawal EW-ES problem is then posed as determining
$J(w, t_0)$ which is given by
\begin{eqnarray}
  \hat{J}(w,t_0) & = & \sup_{W^*} \hat{V}( w, W^*, t_0) ~. \label{J_cont}
\end{eqnarray} 
Following the usual arguments  we obtain the Hamilton-Jacobi-Bellman PDE for $\hat{V}$ 
\begin{eqnarray}
  & & \hat{V}_t + \sup_{ \hat{p} \in \hat{\mathcal{Z}}_p}
                       \sup_{ \hat{q} \in \hat{\mathcal{Z}}_q} 
 \biggl\{
      w \biggl[ \hat{p} \mu^s +(1-\hat{p}) \mu^b \biggr] \hat{V}_w -\hat{q} \hat{V}_w  + \hat{q} 
      + w^2 \biggl[ \frac{ (\hat{p} \sigma^s)^2}{2} + (1-\hat{p})\hat{p} \sigma^s \sigma^b \rho_{sb} 
            + \frac{ ( (1-\hat{p}) \sigma^b)^2}{2} \biggr]  \hat{V}_{ww} \biggr\} \nonumber \\
 & &  = 0 ~,
     \label{cont_hjb}
\end{eqnarray}
with terminal condition
\begin{eqnarray}
  \hat{V}(w, T, W^*) & = & 
         \kappa
             \biggl(
                  W^* + \frac{1}{\alpha} \min((W_T -W^*),0)
               \biggr) ~.
\end{eqnarray}
In general, we seek the viscosity solution of equation (\ref{cont_hjb}), which does not
require that the solution $\hat{V}$ be differentiable.  However, we make the assumption
that $\hat{V}_w$ exists and is bounded.

Rewriting equation (\ref{cont_hjb}) we have
\begin{eqnarray}
  & & \hat{V}_t + \sup_{ \hat{p} \in \hat{\mathcal{Z}}_p}
 \biggl\{
      w \biggl[ \hat{p} \mu^s +(1-\hat{p}) \mu^b \biggr] \hat{V}_w  
      + w^2 \biggl[ \frac{ (\hat{p} \sigma^s)^2}{2} + (1-\hat{p})\hat{p} \sigma^s \sigma^b \rho_{sb} 
       + \frac{ ( (1-\hat{p}) \sigma^b)^2}{2} \biggr]  \hat{V}_{ww} \biggr\} \nonumber \\
    & & ~~~~~~~~+ \sup_{ \hat{q} \in \hat{\mathcal{Z}}_q} \biggl\{
                                       \hat{q} (1 -\hat{V}_w) 
                                              \biggr\}  = 0 ~,
     \label{cont_hjb_2}
\end{eqnarray}
and therefore the
optimal value of  $\hat{q}$ is determined by maximizing
\begin{eqnarray}
      \sup_{ \hat{q} \in \hat{\mathcal{Z}}_q}  \hat{q} (1 -\hat{V}_w) ~.\label{opt_q_cont}
\end{eqnarray}
Breaking ties by choosing $\hat{q} = \hat{q}_{\min}$ if $ (1-\hat{V}_w) =0$, we then have that
the optimal strategy $\hat{q}^*$ is
\begin{eqnarray}
      \hat{q}^* & = & \begin{cases}
                        \hat{q}_{\min} ~;& (1-\hat{V}_w) \leq 0 \\
                         \hat{q}_{\max} ~;& (1-\hat{V}_w) > 0 
                      \end{cases}
             ~. \label{bang_bang_1}
\end{eqnarray} 
Equation \ref{bang_bang_1} holds for any $W^*$ and hence is also true for the optimal
value of  $W^*$ in equation (\ref{J_cont}).
We obtain the same result (after some algebraic complexity) if we assume that the stock
and bond processes are given in equation (\ref{jump_process_stock})
and  equation (\ref{jump_process_bond}).  
\end{proof}

\begin{remark}[Bang-bang control for discrete rebalancing/withdrawals]
\label{quasi_prop}
Proposition \ref{bang_bang_prop} suggests that, for sufficiently 
small rebalancing intervals, we can expect the
optimal $q$ control (finite withdrawal amount) to be bang-bang.  However, it is not clear that this will continue to
be true for the case of annual rebalancing (which we specify in our numerical examples).  In fact,
we do observe that the $q$ control is very close to bang-bang in our numerical experiments,
even for annual rebalancing. We term this control to be {\em quasi-bang-bang}. 
\end{remark}

\subsection{Numerical algorithm: $(PCEE_{t_0}(\kappa))$}\label{control_algo}

\subsubsection{Solution of auxiliary problem}
We solve the auxiliary problem (\ref{expanded_1}-\ref{expanded_2}), with
a fixed values of $W^*$, $\kappa$ and $\alpha$.
We do not allow shorting of stock, so  the amount in the stocks $S(t) \geq 0$.  We discretize
the state space in $s > 0$ using $n_{\hat{x}}$ equally spaced nodes in the $\hat{x} = \log s$ direction, on a finite
localized domain $ s \in [ e^{\hat{x}_{\min}}, e^{ \hat{x}_{\max}}]$.  
We discretize the state space in $b >0$ using $n_y$ equally spaced nodes in the
$y= \log b$ direction, on a finite localized domain $ b \in [b_{\min}, b_{\max}] = [e^{y_{\min}}, e^{y_{\max}}]$.  
We also define a $b^{\prime} >0$ grid,  using $n_b$ equally spaced nodes in the
$y^{\prime} = \log b^{\prime}$ direction, on the localized domain with 
$b^{\prime} \in  [b^{\prime}_{\min}, b^{\prime}_{\max}] = [e^{y_{\min}}, e^{y_{\max}}]$.  
The grid $[s_{\min}, s_{\max}] \times [b_{\min}, b_{\max}]$
represents cases where $b \geq0$.  The grid $[s_{\min}, s_{\max}] \times [b_{\min}^{\prime}, b_{\max}^{\prime}]$
represents cases where $b = -b^{\prime} < 0$.

We use  the Fourier methods discussed in \citet{forsythlabahn2017} to solve PIDE (\ref{expanded_7})
between rebalancing times.  Further details concerning the Fourier method can be found in \citet{forsyth_2020_a}.

We choose the localized domain $[\hat{x}_{\min}, \hat{x}_{\max}] = [\log(10^2) -8, \log(10^2)+8]$,
with $[y_{\min}, y_{\max}] = [\hat{x}_{\min}, \hat{x}_{\max}]$ (units thousands of dollars).  Wrap-around
effects are minimized using the domain extension method in \citet{forsythlabahn2017}.   In our numerical experiments,
we carried out tests replacing $[\hat{x}_{\min}, \hat{x}_{\max}]$ by $[\hat{x}_{\min} -2, \hat{x}_{\max} +2]$
and similarly replacing $[y_{\min}, y_{\max}]$ by $[y_{\min} -2, y_{\max} +2]$.  In all cases,
this resulted in changes to the summary statistics in at most the fifth digit, verifying that
the localization error is small.

We discretize the $p$ controls using an equally spaced grid with $n_y$ values.
We then solve the optimization problem (\ref{expanded_3}) using exhaustive search over
the discretized $p$ values, linearly interpolating the right hand side 
discrete values of $V$  in equation (\ref{expanded_3})
as required.  We store the optimal control for $p$ at $n_y$ discrete
wealth nodes.    We also discretize the controls for $q$ in the range $[q_{\min}, q_{\max}] $
in increments of one thousand dollars, and determine the optimal control for $q$
by exhaustive search.  
We then determine the optimal control for $q$ using equation (\ref{q_opt}), at a set
of $n_y$ discrete $w$ nodes.  We use the previously stored controls for $p$ in order
to evaluate the right hand side of equation (\ref{q_opt}), linearly interpolating the controls
if necessary.

We use a fixed discretization of the $q$ controls
since it is realistic to assume that retirees will change withdrawal amounts in fairly coarse increments.
As we shall see, as suggested by Proposition \ref{bang_bang_prop}, the $q$ control turns out
to be quasi-bang-bang, hence the discretization of the $q$ control 
hardly makes any difference to the solution.  

Finally, stored controls for $q$ and $p$ are then used to advance the solution in equation (\ref{expanded_4}), 
linearly interpolating
the controls and value function if required.

Assume that $n_{\hat{x}} = O(n_y)$.  Then, the cost of using an FFT method to solve equation (\ref{expanded_7})
between rebalancing times is $O( n_y^2 \log n_y)$.  The cost of determining the optimal control
for $p$ in using equation (\ref{expanded_4}) at $n_y$ discrete $w$ values, using exhaustive search,
is $O(n_y^2)$.  The cost of determining the optimal $q$ using equation (\ref{q_opt}) at $n_y$
discrete $w$ values is $O(n_y)$, since the number of discrete $q$ controls is $O(1)$.
In addition, the step (\ref{expanded_4}) has complexity $O(n_y^2)$.
The total number of rebalancing times is fixed, hence the total complexity of the
solution of problem (\ref{expanded_1}) for a fixed value of $W^*$ is $O(n_y^2 \log n_y)$.

\subsubsection{Outer optimization over $W^*$}
Given an approximate solution of the auxiliary problem (\ref{expanded_1}-\ref{expanded_2}) 
at $t=0$, which we denote by $V(s, b, W^*, 0)$, we then determine the final solution
for problem  $PCEE_{t_0}(\kappa))$  in equations (\ref{PCEE_a}-\ref{PCEE_b}) using
equation (\ref{expanded_equiv}).  More specifically, we solve
\begin{eqnarray}
    J( 0, W_0, 0^-) & = & \sup_{W^{\prime}} V(0, W_0, W^{\prime}, 0^-) \nonumber \\
                    & &  W_0 = {\text{ initial wealth}} ~.
                 \label{outer_opt}
\end{eqnarray}
We solve the auxiliary problem on sequence of grids $n_{\hat{x}} \times  n_y$.  On the coarsest
grid, we discretize $W^*$ and solve problem (\ref{outer_opt}) by exhaustive search.
We use this optimal value of $W^*$ as a starting point to a one dimensional optimization
algorithm on a sequence of finer grids.  

This approach does not guarantee that we have the globally optimal solution to 
problem (\ref{outer_opt}), since the problem is not guaranteed to be convex.  However, we
have made a few tests by carrying out a grid search on the finest grid, which suggest
that we do indeed have the globally optimal solution.

\subsubsection{Stabilization}
If $W_t \gg  W^* $,
and $t \rightarrow T$, then $Pr[ W_T < W^*] \simeq 0$ (recall that $W^*$ is fixed for
problem $\left(\mathit{TCEQ_{t_n}}\left(\kappa / \alpha \right)\right)$ (\ref{timec_equiv}) ). In addition, 
for large values of $W_t$, the withdrawal will be capped at $q_{\max}$.
In this fortuitous situation for the retiree,
the control only weakly effects the objective function.  To avoid
this ill-posedness for the controls,
we
changed the objective function (\ref{PCEE_a}) to
\begin{eqnarray}
    \qquad J\left(s,b,  t_{0}^-\right)
    & = & \sup_{\mathcal{P}_{0}\in\mathcal{A}}
          \sup_{W^*}
        \Biggl\{
               E_{\mathcal{P}_{0}}^{X_0^+,t_{0}^+}
           \Biggl[ ~\sum_{i=0}^{i=M} q_i ~  + ~
              \kappa \biggl( W^* + \frac{1}{\alpha} \min (W_T -W^*, 0) \biggr)
                  \overbrace{+ \epsilon W_T}^{stabilization}
                    \Biggr. \Biggr.  \nonumber \\
         & &  \Biggl. \Biggl. ~~~~~~~~~~~~~~~
                \bigg\vert X(t_0^-) = (s,b)
                   ~\Biggr] \Biggr\} ~. \label{stable_objective}
\end{eqnarray}
We used the value $\epsilon = +10^{-6}$ in the following test cases.  
Using a positive value for $\epsilon$ has the effect of forcing the
strategy to invest in stocks when $W_t$ is very large, and $t \rightarrow T$,
when the control problem is ill-posed.
In other words, when the probability  that $W_T$ is less than $W^*$ is very small, then the ES risk
is practically zero, hence the investor might as well invest in risky assets.
There is little to lose, and much to gain (at least for the retiree's estate).
Note that
using this small value of $\epsilon = 10^{-6}$ gave the same results as
$\epsilon = 0$ for the summary statistics, to four digits.  This is simply because
the states with very large wealth have low probability.  However, this stabilization
procedure produced more smooth heat maps for large wealth values, without altering
the summary statistics appreciably.

\section{Data}\label{data_section}

We use data from the Center for
Research in Security Prices (CRSP) on a monthly basis over the
1926:1-2019:12 period.\footnote{More specifically, results presented here
were calculated based on data from Historical Indexes, \copyright
2020 Center for Research in Security Prices (CRSP), The University of
Chicago Booth School of Business. Wharton Research Data Services was
used in preparing this article. This service and the data available
thereon constitute valuable intellectual property and trade secrets
of WRDS and/or its third-party suppliers.} Our base case tests use
the CRSP 10 year US treasury index for the bond asset
and the CRSP value-weighted total return index for the stock asset.
This latter index includes all distributions for all domestic stocks
trading on major U.S.\ exchanges.\footnote{The
10-year Treasury index was constructed from monthly returns from CRSP
back to 1941. The data for 1926-1941 were interpolated from annual
returns in \citet{Homer_rates}.} All of these various indexes are in
nominal terms, so we adjust them for inflation by using the U.S.\ CPI
index, also supplied by CRSP. We use real indexes since investors saving
for retirement should be focused on real (not nominal) wealth goals.

We use the threshold technique \citep{mancini2009,contmancini2011,Dang2015a}
to estimate the parameters for the parametric stochastic process models.
Note that the data is inflation adjusted, so that all parameters
reflect real returns.
Table \ref{fit_params} shows the results of calibrating
the models to the historical data.  
The correlation $\rho_{sb}$
is computed by removing any returns which occur at times corresponding
to jumps in either series, and then using the sample covariance.
Further discussion of the validity of assuming that the stock and
bond jumps are independent is given in \citet{forsyth_2020_a}.

{\small
\begin{table}[hbt!]
\begin{center}
\begin{tabular}{cccccccc} \toprule
 CRSP & $\mu^s$ & $\sigma^s$ & $\lambda^s$ & $p_{\text{\emph{up}}}^s$ &
  $\eta_1^s$ & $\eta_2^s$ & $\rho_{sb}$ \\ \midrule
       & 0.0877  & 0.1459&  0.3191  &  0.2333 & 4.3608 & 5.504 & 0.04554\\
 \midrule
 \midrule
10-year Treasury & $\mu^b$ & $\sigma^b$ & $\lambda^b$ & $p_{\text{\emph{up}}}^b$ &
  $\eta_1^b$ & $\eta_2^b$ & $\rho_{sb}$ \\ \midrule
        & 0.0239 & 0.0538 & 0.3830 &  0.6111 &  16.19 & 17.27 & 0.04554\\
\bottomrule
\end{tabular}
\caption{Estimated annualized parameters for double exponential jump
diffusion model.  Value-weighted CRSP index, 10-year US treasury index
deflated by the CPI.  Sample period 1926:1 to 2019:12. 
}
\label{fit_params}
\end{center}
\end{table}
}

\section{Investment scenario}
Table \ref{base_case_1} shows our base case investment scenario. 
We will use thousands as our units of wealth in the following.  For example,
a withdrawal of $40$ per year corresponds to $\$40,000$ per
year, with an initial wealth of $1000$ ($\$1,000,000$).    Thus, a withdrawal of 40 per year
would correspond to the use of the four per cent rule \citep{Bengen1994}.

To make this example more concrete, this scenario would apply to a
retiree who is $65$ years old, with a pre-retirement salary of \$100,000 per year,
with a total value of DC plan holdings at retirement of \$1,000,000.
In Canada, a retiree would be eligible for government benefits (indexed)
of about \$20,000 per year.  If the investor targets withdrawing \$40,000 per year
from the DC plan, then this would result in  total real income of about \$60,000 per year,
which is about 60\% of pre-retirement salary.
For risk management purposes, we will assume that the retiree owns mortgage free real estate worth
about \$400,000, which will retain its value in real terms over 30 years. If our measure
of risk is Expected Shortfall  at the 5\% level, then we suppose that any ES  which is greater
than about -\$200,000 (the negative of one half the value of the real estate)
can be managed using a reverse mortgage. 

Note that in Table \ref{base_case_1} we have set the borrowing spread $\mu_c^{b} =0$.
The (real) drift rate of the 10-year treasury index is about 200 bps larger than the
30-day T-bill index.  Hence, borrowing at the 10-year treasury rate is roughly comparable
to borrowing at the short term rate plus a spread of about 200 bps, which we suppose to
be a reasonable  estimate for a well secured reverse mortgage.

\begin{table}[hbt!]
\begin{center}
\begin{tabular}{lc} \toprule
Investment horizon $T$ (years) & 30  \\
Equity market index & CRSP Cap-weighted index (real) \\
Bond index & 10-year Treasury (US) (real) \\
Initial portfolio value $W_0$  & 1000 \\
Cash withdrawal times & $t=0,1,\ldots, 30$\\
Withdrawal range   & $[q_{\min}, q_{\max}]$\\
Equity fraction range & $[0,1]$\\
Borrowing spread $\mu_c^{b}$ & 0.0 \\
Rebalancing interval (years) & 1  \\
Market parameters & See Table~\ref{fit_params} \\ \bottomrule
\end{tabular}
\caption{Input data for examples.  Monetary units: thousands of dollars.
\label{base_case_1}}
\end{center}
\end{table}

\subsection{Synthetic market}
We fit the parameters for the parametric stock and bond processes
(\ref{jump_process_stock} - \ref{jump_process_bond}) as described in 
Section ~\ref{data_section}.  We then compute
and store the optimal controls based on the parametric market model.
Finally, we compute various statistical quantities by using the stored control,
and then carrying out Monte Carlo simulations, based on processes (\ref{jump_process_stock} - \ref{jump_process_bond}).

\subsection{Historical market}\label{boot_section}
We compute and store the optimal controls based on the parametric
model (\ref{jump_process_stock}-\ref{jump_process_bond}) as for the
synthetic market case.  However, we compute statistical quantities
using the stored controls, but using bootstrapped historical return
data directly.  We remind the reader that all returns are
inflation adjusted.  We use the stationary block bootstrap method \citep{politis1994,politis2004,politis2009,dichtl2016}.
A crucial parameter is the expected blocksize. Sampling the data in blocks
accounts for serial correlation in the data series.    We use the algorithm in \citet{politis2009} to determine
the optimal blocksize for the bond and stock returns separately, see Table \ref{auto_blocksize}. 
We use a paired sampling approach to
simultaneously draw returns from both time series.  In this case,  a reasonable
estimate for the blocksize for the paired resampling algorithm would
be about $.25$ years.
We will give results for a range of blocksizes as a check on the robustness of the bootstrap
results.
Detailed pseudo-code for block bootstrap resampling is given in \citet{Forsyth_Vetzal_2019a}.

\begin{table}[tb]
\begin{center}
{\small
\begin{tabular}{lc} \toprule
Data series          & Optimal expected \\
                     & block size $\hat{b}$ (months) \\ \midrule
  Real 10-year Treasury index    &  4.2 \\
  Real CRSP value-weighted index &  3.1 \\
\bottomrule
\end{tabular}
}
\end{center}
\caption{Optimal expected blocksize $\hat{b}=1/v$ when the blocksize follows
a geometric distribution $Pr(b = k) = (1-v)^{k-1} v$. The algorithm in
\citet{politis2009} is used to determine $\hat{b}$.
Historical data range 1926:1-2019:12.
\label{auto_blocksize}
}
\end{table}

\section{Synthetic and historical markets: constant withdrawals $q=40$, constant proportion strategy}
\label{constq_constp}
We consider the scenario in Table \ref{base_case_1}.
As a benchmark, we consider withdrawing at a constant rate of $40$ per year
(units: thousands of dollars).  This would correspond to the 4\% rule suggested
in \citep{Bengen1994}.  We also assume that the portfolio is rebalanced
to a constant weight in stocks each year.  Table \ref{const_wt_const_withdraw_40_day}
shows the results for various equity weights in the synthetic market,
while Table \ref{const_wt_const_withdraw_boot} shows results for the
bootstrapped historical market.

Note that the results are roughly comparable for both synthetic and
historical markets.  However, none of the cases with constant withdrawals
and constant equity weights meets our criteria of an ES $> - \$200,000$.

\begin{table}[hbt!]
\begin{center}
\begin{tabular}{ccc} \toprule
Equity Weight & Expected Shortfall (5\%)  & $Median[W_T]$     \\ \midrule
   0.0        & -469.4                   & 127.4  \\
   0.2        & -288.6                   & 579.3  \\
   0.4        & -295.5                   & 1137 \\
   0.6        & -436.0                    & 1762 \\
   0.8        & -630.6                    & 2374 \\
\bottomrule
\end{tabular}
\caption{Synthetic market results assuming
the scenario given in Table~\ref{base_case_1}, with $q_{\max} = q_{\min} = 40$,
and $p_{\ell} = constant$ in equation (\ref{PCEE_b}). Stock index: real capitalization weighted CRSP stocks;
bond index: real 10-year US treasuries.  Parameters from Table \ref{fit_params}.
Units: thousands of dollars. Statistics
based on $2.56 \times 10^6$ Monte Carlo simulation runs.
\label{const_wt_const_withdraw_40_day}
}
\end{center}
\end{table}

\begin{table}[hbt!]
\begin{center}
\begin{tabular}{ccc} \toprule
Equity Weight & Expected Shortfall (5\%)  & $Median[W_T]$     \\ \midrule
   0.0        &-444.1                   & 132.6 \\
   0.2        & -278.6                   & 555.2\\
   0.4        & -267.5                &  1083 \\
   0.6        & -374.5                &  1694  \\
   0.8        & -537.7                & 2322 \\
\bottomrule
\end{tabular}
\caption{Historical market results (bootstrap resampling) assuming
the scenario given in Table~\ref{base_case_1}, except that $q_{\max} = q_{\min} = 40$,
and $p_{\ell} = constant$ in equation (\ref{PCEE_b}). Stock index: real capitalization weighted CRSP stocks;
bond index: real 10-year US treasuries.  Historical data in range 1926:1-2019:12.
Parameters from Table \ref{fit_params}.
Units: thousands of dollars. Statistics
based on $10^5$ bootstrap resampling simulations.  Expected blocksize $0.25$ years.
\label{const_wt_const_withdraw_boot}
}
\end{center}
\end{table}

\section{Synthetic market}

\subsubsection{Convergence test: synthetic market}
We carry out an initial test of convergence of our numerical method
for the EW-ES problem (\ref{PCEE_a}).  Table \ref{conservative_accuracy} shows the results
for solution of the PDE on a sequence of grids.  For each refinement level, we store the optimal
control, and use this control in Monte Carlo simulations.   The PDE solution appears to
converge at roughly a first order rate.  However, the Monte Carlo simulations (based on
the PDE controls) appear to be slightly more accurate.  This effect has also been noted
in \citet{Ma2015}.  In the following, we will report results based on (i) determining the
control from the PDE solution (using the finest grid in Table \ref{conservative_accuracy})
and (ii) using this control in Monte Carlo simulations.

\begin{table}[hbt!]
\begin{center}
\begin{tabular}{lcc|cc} \toprule
\multicolumn{3}{c|}{Algorithm in Section \ref{algo_section} } & \multicolumn{2}{c}{Monte Carlo} \\ \midrule
Grid &   ES (5\%) & $E[ \sum_i q_i]/(M+1)$  &   ES (5\%) & $E[ \sum_i q_i]/(M+1)$  \\
\midrule 
$512 \times 512$ & -16.788     & 49.7470    & -5.035      &  50.35\\
$1024 \times 1024$ & -9.3609    &  49.8513   &  -4.511    & 49.86 \\
$2048 \times 2048$ & -7.6954   &   49.8998   &  -4.732  &  49.89 \\
\midrule
\bottomrule
\end{tabular}
\caption{Convergence test,
real stock index: deflated real capitalization weighted CRSP, real bond index: deflated
ten year Treasuries.  Scenario in Table \ref{base_case_1}.
Parameters in Table \ref{fit_params}.
The Monte Carlo method used
$2.56 \times 10^6$ simulations.
$\kappa = 1.0, \alpha = .05$.
Grid refers to the grid
used in the Algorithm in Section \ref{algo_section}: $n_x \times n_b$, where $n_x$ is
the number of nodes in the $\log s$ direction, and
$n_b$ is the number of nodes in the $\log b$ direction.
Units: thousands of dollars (real).
$(M+1)$ is the total number of withdrawals. $M$ is the number of rebalancing dates.
$q_{\min} = 35.0$. $q_{\max} = 60$.
$W^* = 204.6$ (equation(\ref{PCEE_a})) on the finest grid, Algorithm in Section \ref{algo_section}.
\label{conservative_accuracy}}
\end{center}
\end{table}

\subsection{Constant withdrawals}
As a benchmark strategy, we solve problem (\ref{PCEE_a}), scenario in Table \ref{base_case_1},
but force a constant withdrawal, i.e. we set $q_{\min} = q_{\max}$, but retain the
optimal asset allocation control.  
The results are shown in Table \ref{optimal_p_q_const}.
We can see from Table \ref{optimal_p_q_const}  that constant withdrawals of $35$ and $40$ per year meet
our objective that $ES > -200$ (recall that units are thousands of dollars).
The strategy of withdrawing $40$ per year, coupled with optimal asset allocation, is a reasonable
strategy, which meets both our income and risk targets.  However, note that
$Median[W_T] = 717$, indicating that 50\% of the time, we leave over
\$700,000 on the table at the end our investment horizon.  In other words,
the constant withdrawal rate of \$40,000 per year, while being reasonably 
safe over 30 years, paradoxically also has a high probability of underspending.
This leads us to then allow the additional flexibility of variable spending.

\begin{table}[hbt!]
\begin{center}
\begin{tabular}{cccc} \toprule
$q_{\max} = q_{\min}$&  ES (5\%) &    $Median[W_T]$ & $\sum_i Median(p_i) /M$     \\ \midrule
                  35 & 31.03     & 952.2   & .271 \\
                  40 &  -196.1 & 716.6   & .357 \\
                  45 &  -425.4 & 441.4   & .424 \\
\bottomrule
\end{tabular}
\caption{Synthetic market results for optimal strategies,  assuming
the scenario given in Table~\ref{base_case_1}. Stock index: real capitalization weighted CRSP stocks;
bond index: real 10 year US treasuries.  Parameters from Table \ref{fit_params}.
Units: thousands of dollars. Statistics
based on $2.56 \times 10^6$ Monte Carlo simulation runs.
Control is computed using the Algorithm in Section \ref{algo_section}, stored, and then used
in the Monte Carlo simulations.
$(M+1)$ is the number of withdrawals.
$M$ is the number of rebalancing dates.
$\epsilon = 10^{-6}$.
\label{optimal_p_q_const}
}
\end{center}
\end{table}

\subsection{Synthetic market: efficient frontiers}
We solve problem (\ref{PCEE_a}), scenario in Table \ref{base_case_1},
and now  allow the withdrawal to be determined from our optimal strategy.
We compute the efficient EW-ES frontiers for two cases: $[q_{\min},q_{\max}] = [35,60]$
and $[q_{\min},q_{\max}] = [40,65]$, as shown  in Figure \ref{EW_ES_frontiers}.  We also show
the single points corresponding to constant withdrawals (from Table \ref{optimal_p_q_const}
for $q=35,40$) on the Figures as well.  Detailed tables showing statistics for each point
on the efficient frontier are given in Tables \ref{optimal_p_q} and \ref{optimal_p_q_2}.

For sufficiently large $\kappa$ we expect that the the efficient frontier should
converge to the constant withdrawal with $q = q_{\min}$.  However, numerically,
we were not able to obtain accurate solutions for very large values of $\kappa$,
hence the efficient frontiers are shown as ending above the constant withdrawal
points in Figure \ref{EW_ES_frontiers}. The dotted lines represent the extrapolated
values of the efficient frontiers.  Note that these dotted lines are almost vertical,
indicating that very small decreases in ES cause very large changes in EW.  This is, of
course, why it is hard to track the curve (numerically) along these points.

Both of these efficient frontiers are qualitatively similar, so we focus on 
Figure \ref{frontier_40_65}.  Compare the variable withdrawal strategy to the
fixed withdrawal strategy.  The fixed withdrawal strategy $q=40$, from Table \ref{optimal_p_q_const},
has $ES = -196.1$.
If we pick the point on the EW-ES curve with expected average withdrawals of
$53.4$, this corresponds to an $ES = -199.8$ (from Table \ref{optimal_p_q_2}).  In other words, by accepting
a very small amount of extra risk (a smaller ES), we have a strategy
which never withdraws less than $40$ per year, but on average withdraws
$53.4$ per year.  At first sight, this seems to be a very counterintuitive result.
However, from Table \ref{optimal_p_q_const}, we can see that $Median[W_T] = 717$ for
constant $q=40$, while from Table \ref{optimal_p_q_2}, the point $(EW,ES) = (53.4, -199.8)$
has $Median[W_T] = 78.3$.  This means that the optimal variable withdrawal strategy
is simply much more efficient in withdrawing cash over the 30 year  horizon,
in the event the investments do well.

\begin{figure}[tb]
\centerline{%
\begin{subfigure}[t]{.40\linewidth}
\centering
\includegraphics[width=\linewidth]{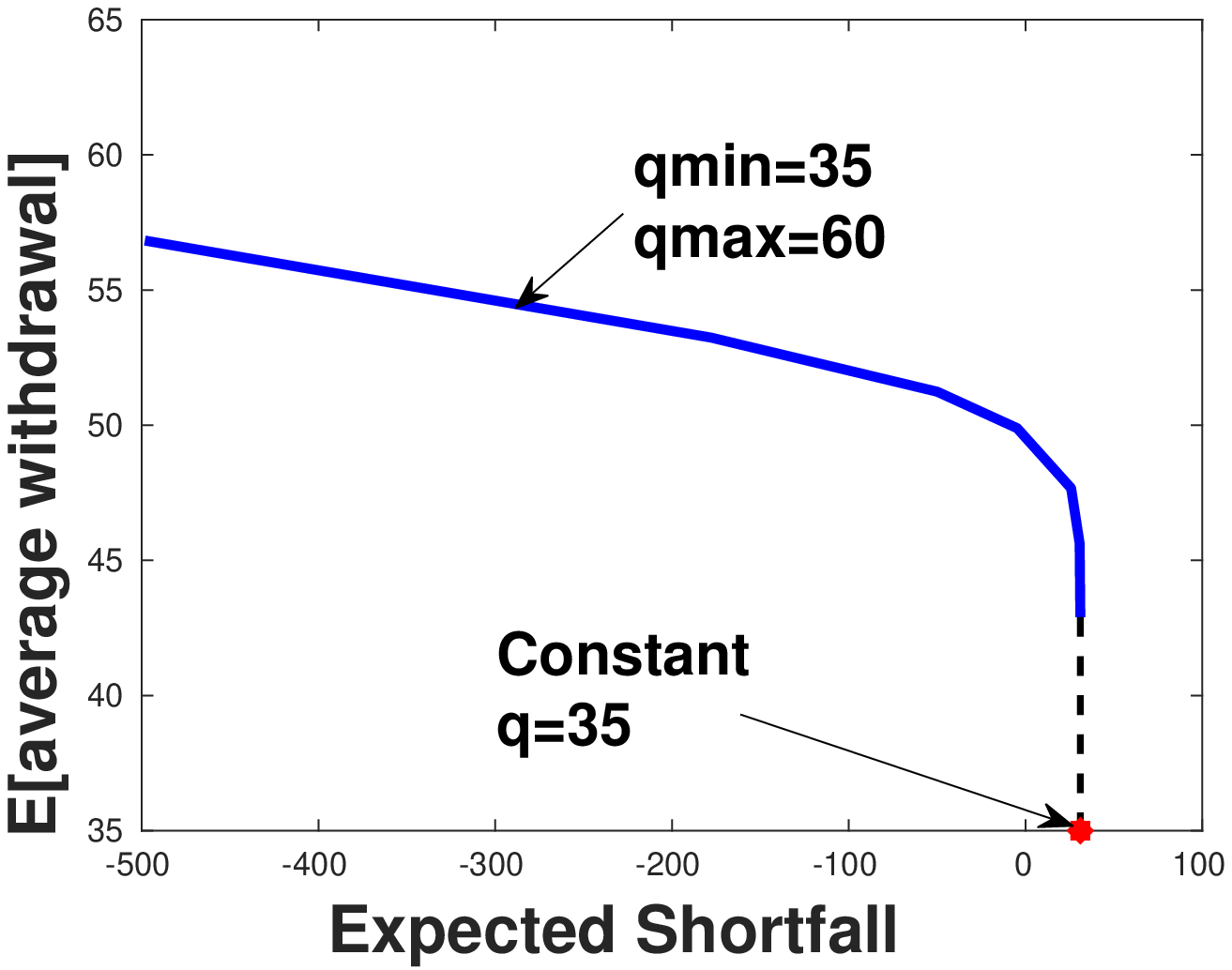}
\caption{$q_{\min}=35, q_{\max} = 60$.}
\label{frontier_35_60}
\end{subfigure}
\begin{subfigure}[t]{.40\linewidth}
\centering
\includegraphics[width=\linewidth]{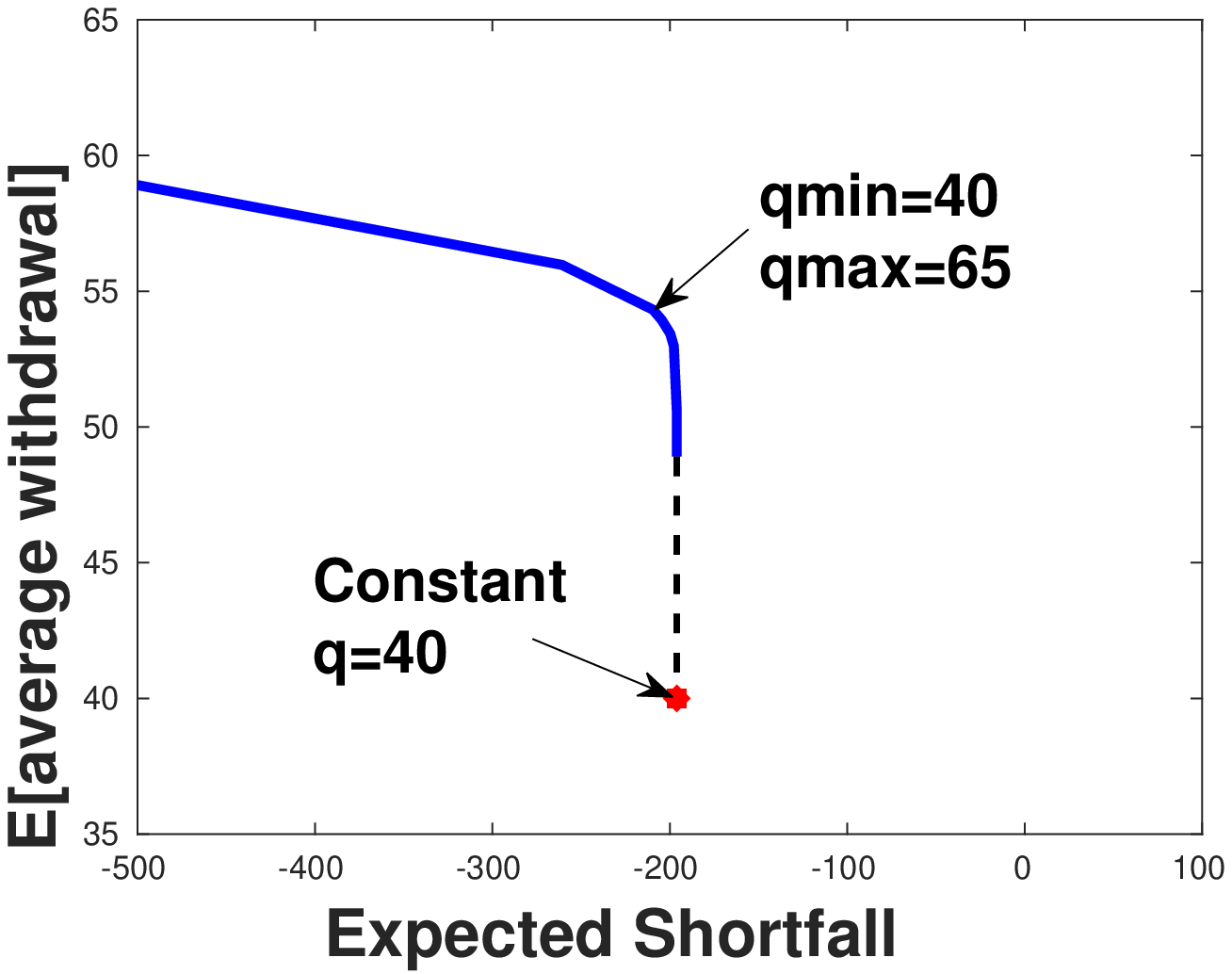}
\caption{$q_{\min}=40, q_{\max} = 65$.}
\label{frontier_40_65}
\end{subfigure}
}
\caption{EW-ES frontiers. Scenario in Table \ref{base_case_1}.
Optimal control computed from problem (\ref{PCEE_a}),
Parameters based on real CRSP index,
real 10-year US treasuries  (see Table \ref{fit_params}).  Control computed and stored from
the PDE solution (synthetic market).
Frontier computed using 
$2.56 \times 10^{6}$ MC simulations.
Units: thousands of dollars.
$\epsilon = 10^{-6}$.
}
\label{EW_ES_frontiers}
\end{figure}

\subsection{Synthetic market: optimal controls, withdrawals, wealth and heat map}
The percentiles of fraction in equities, wealth and withdrawals, for the
point on the efficient frontier $(EW,ES) = (51.3, -50.9)$ are shown
in Figure \ref{percentiles_35_60}, for the case $[q_{\min}, q_{\max}] = [35,60]$.
The heat maps of the controls for fraction in equities and optimal withdrawals
are given in Figure \ref{heat_map_fig}.  The normalized withdrawal
is $(q - q_{\min})/(q_{\max} - q_{\min})$.

Note the interesting feature of the median withdrawal in Figure \ref{percentiles_q_35_60}.
The median withdrawal stays at $q=35$ for the first five years of retirement, then
increases rapidly  to $q=60$ by year seven.  This is a result of the fact that
the optimal withdrawal is very close to a {\em bang-bang} type control,
as seen in the heat map shown in Figure \ref{heat_withdrawal_35_60}.  This is not unexpected,
due the fact that in the continuous withdrawal/rebalancing limit, the withdrawal control
(for a rate of withdrawals) is in fact bang-bang, as noted in Proposition \ref{bang_bang_prop}.

The corresponding percentiles and heat maps for the case where $[q_{\min}, q_{\max}] = [40,65]$
are given in Figures \ref{percentiles_40_65} and \ref{heat_map_40_65},
for the point on the EW-ES curve $(EW,ES) = (54.3, -209.5)$.  These figures
are qualitatively similar to the $[q_{\min}, q_{\max}] = [35,60]$ case.

\begin{figure}[tb]
\centerline{%
\begin{subfigure}[t]{.33\linewidth}
\centering
\includegraphics[width=\linewidth]{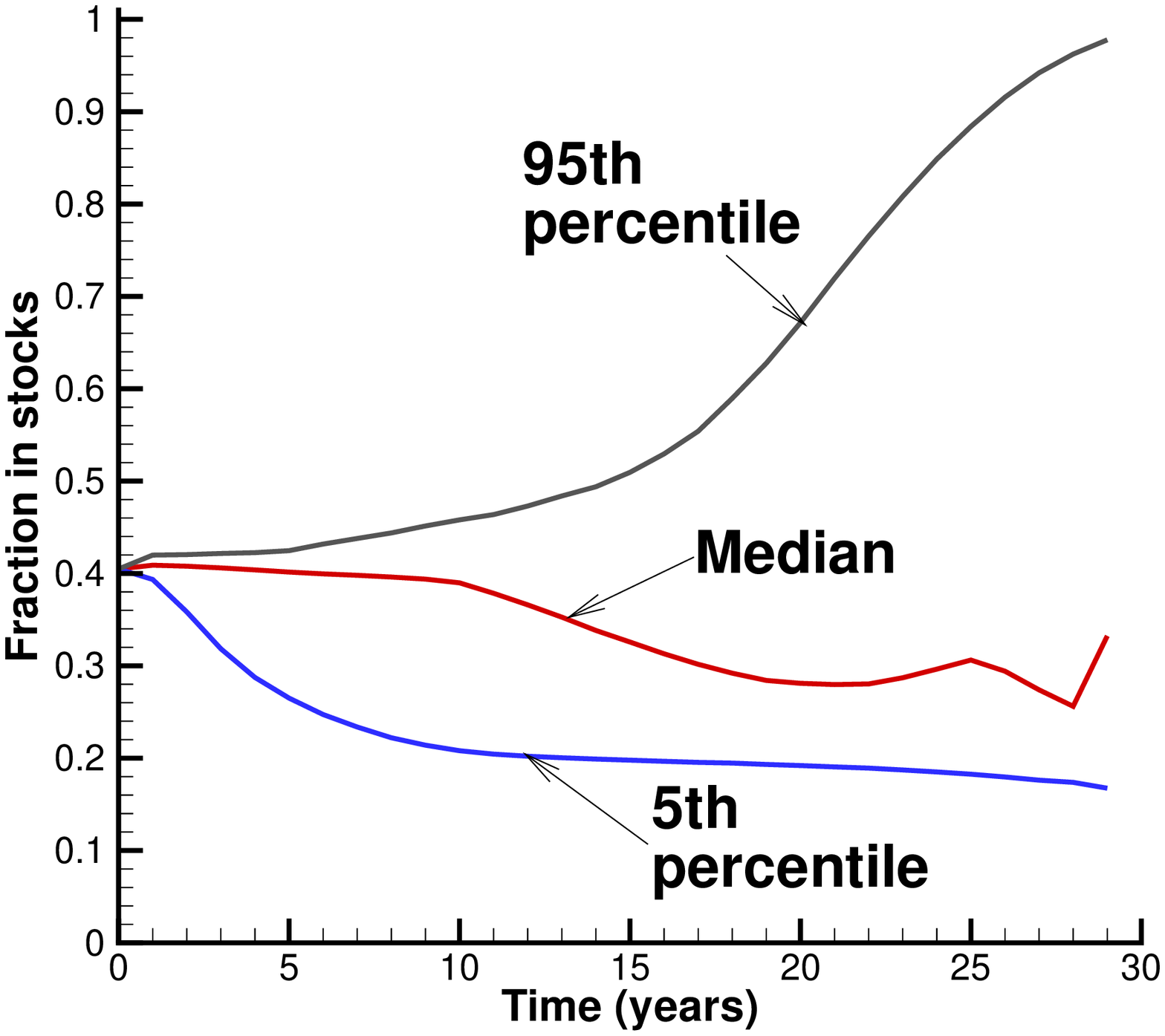}
\caption{Percentiles fraction in stocks}
\label{percentile_stocks_35_60}
\end{subfigure}
\begin{subfigure}[t]{.33\linewidth}
\centering
\includegraphics[width=\linewidth]{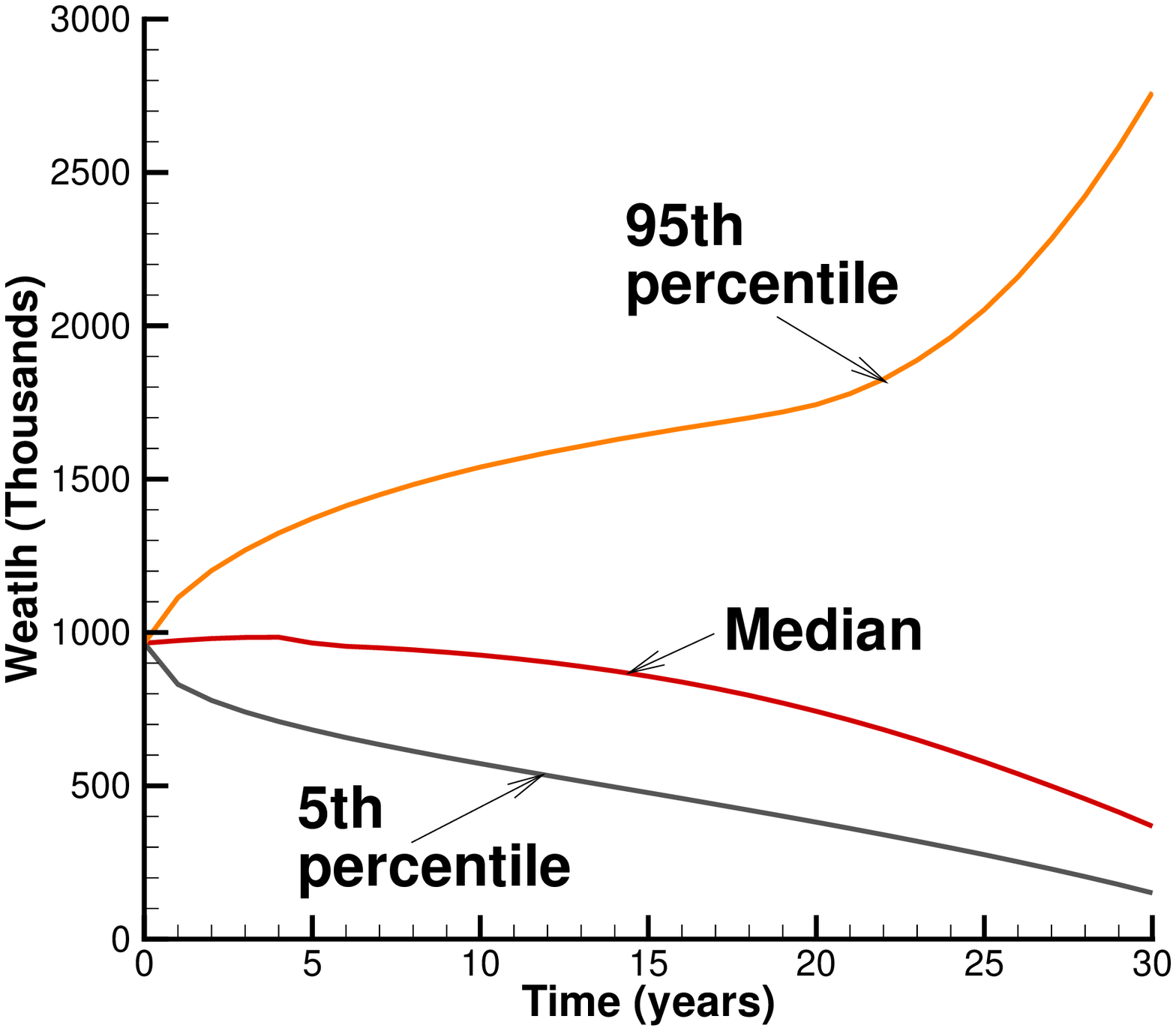}
\caption{Percentiles  wealth}
\label{percentiles_wealth_35_60}
\end{subfigure}
\begin{subfigure}[t]{.33\linewidth}
\centering
\includegraphics[width=\linewidth]{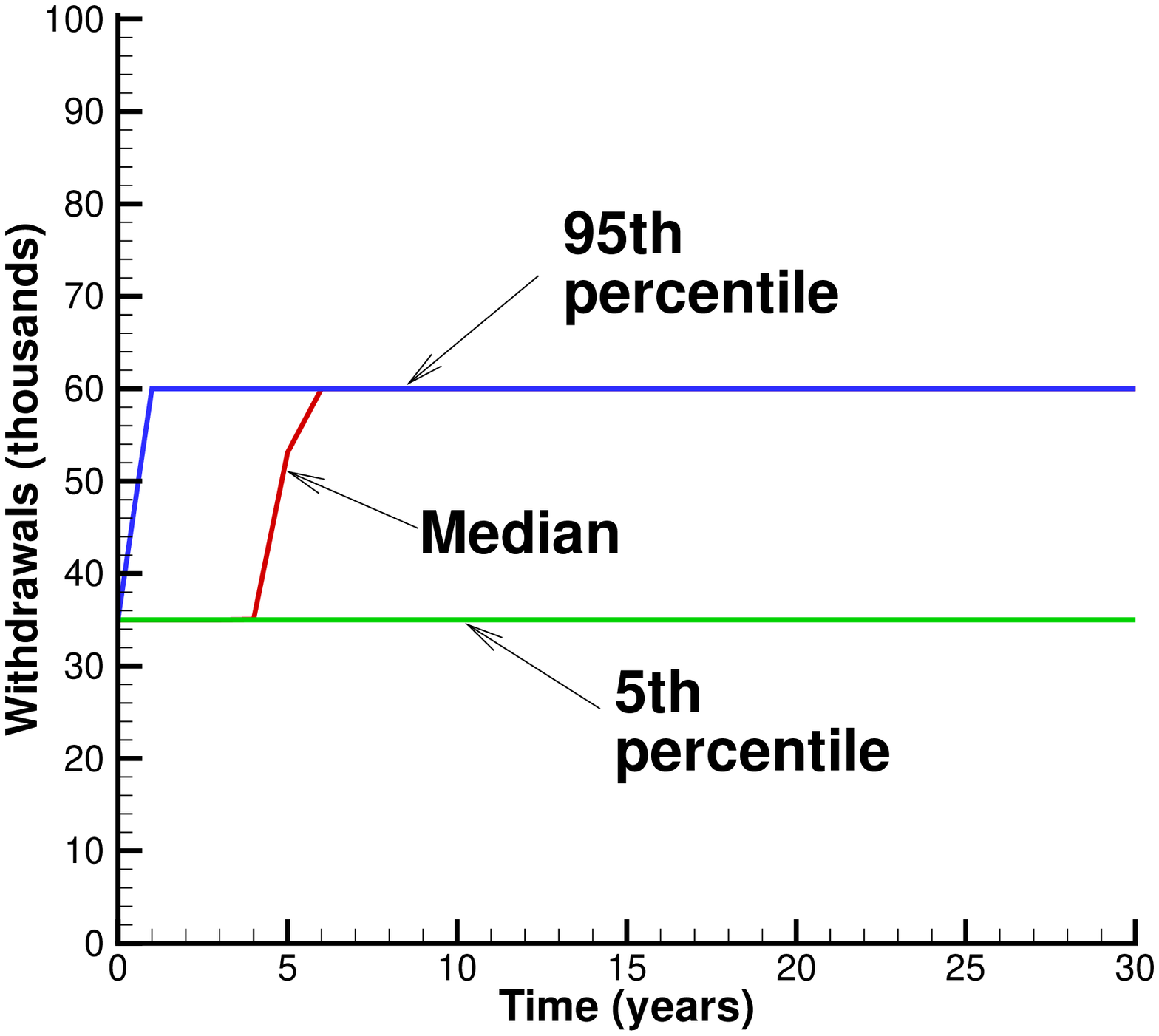}
\caption{Percentiles withdrawals}
\label{percentiles_q_35_60}
\end{subfigure}
}
\caption{Scenario in Table \ref{base_case_1}.
Optimal control computed from problem (\ref{PCEE_a}).
Parameters based on the real CRSP index,
and real 10-year treasuries (see Table \ref{fit_params}).  Control computed and stored
from the PDE solution.
Synthetic market, $2.56 \times 10^{6}$ MC simulations.
$q_{min} = 35, q_{\max} = 60$, $\kappa = 0.5$. $W^* = 177.9$.
$\epsilon = 10^{-6}$.
Units: thousands of dollars.
}
\label{percentiles_35_60}
\end{figure}

\begin{figure}[htb!]
\centerline{
\begin{subfigure}[t]{.4\linewidth}
\centering
\includegraphics[width=\linewidth]{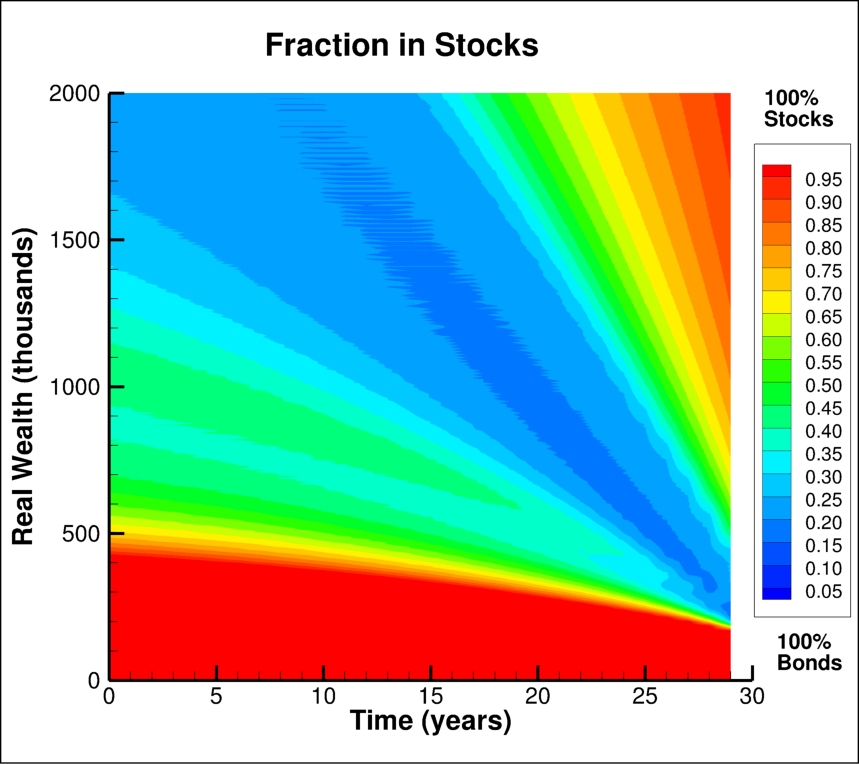}
\caption{Fraction in stocks}
\label{heat_stocks_35_60}
\end{subfigure}
\hspace{.05\linewidth}
\begin{subfigure}[t]{.4\linewidth}
\centering
\includegraphics[width=\linewidth]{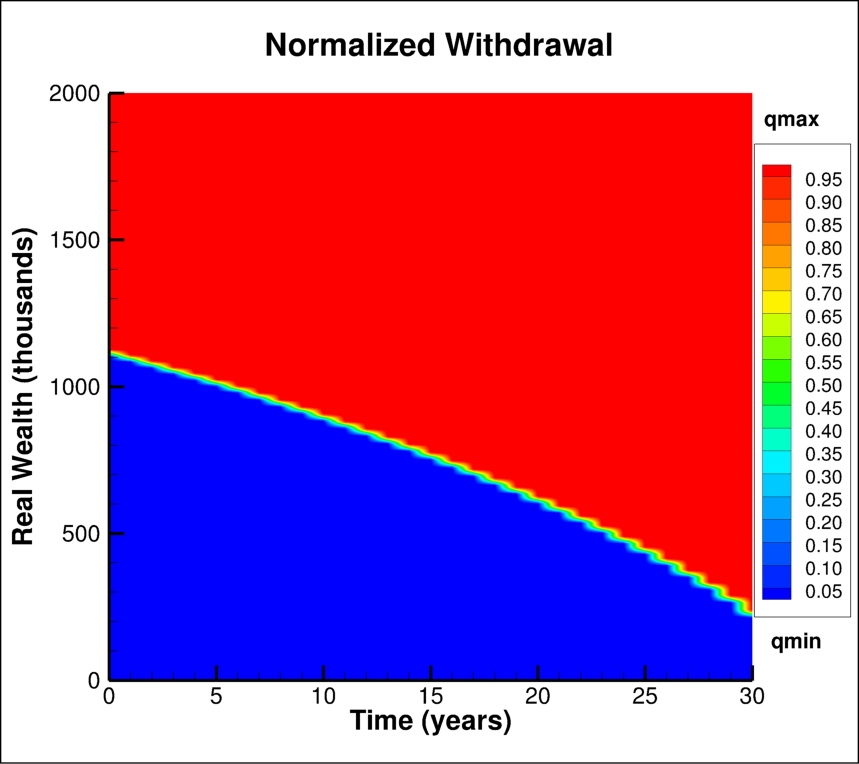}
\caption{Withdrawals}
\label{heat_withdrawal_35_60}
\end{subfigure}
}
\caption{
Heat map of controls: fraction in stocks and
withdrawals, computed from  problem (\ref{PCEE_a}).
cap-weighted real CRSP, real 10 year treasuries.
Scenario given in Table~\ref{base_case_1}.
Control computed and stored from the PDE solution.
$q_{min} = 35, q_{\max} = 60$, $\kappa = 0.5$. $W^* = 177.9$.
$\epsilon = 10^{-6}$.
Normalized withdrawal $(q - q_{\min})/(q_{\max} - q_{\min})$.
Units: thousands of dollars.
\label{heat_map_fig}}
\end{figure}

\begin{figure}[tb]
\centerline{%
\begin{subfigure}[t]{.33\linewidth}
\centering
\includegraphics[width=\linewidth]{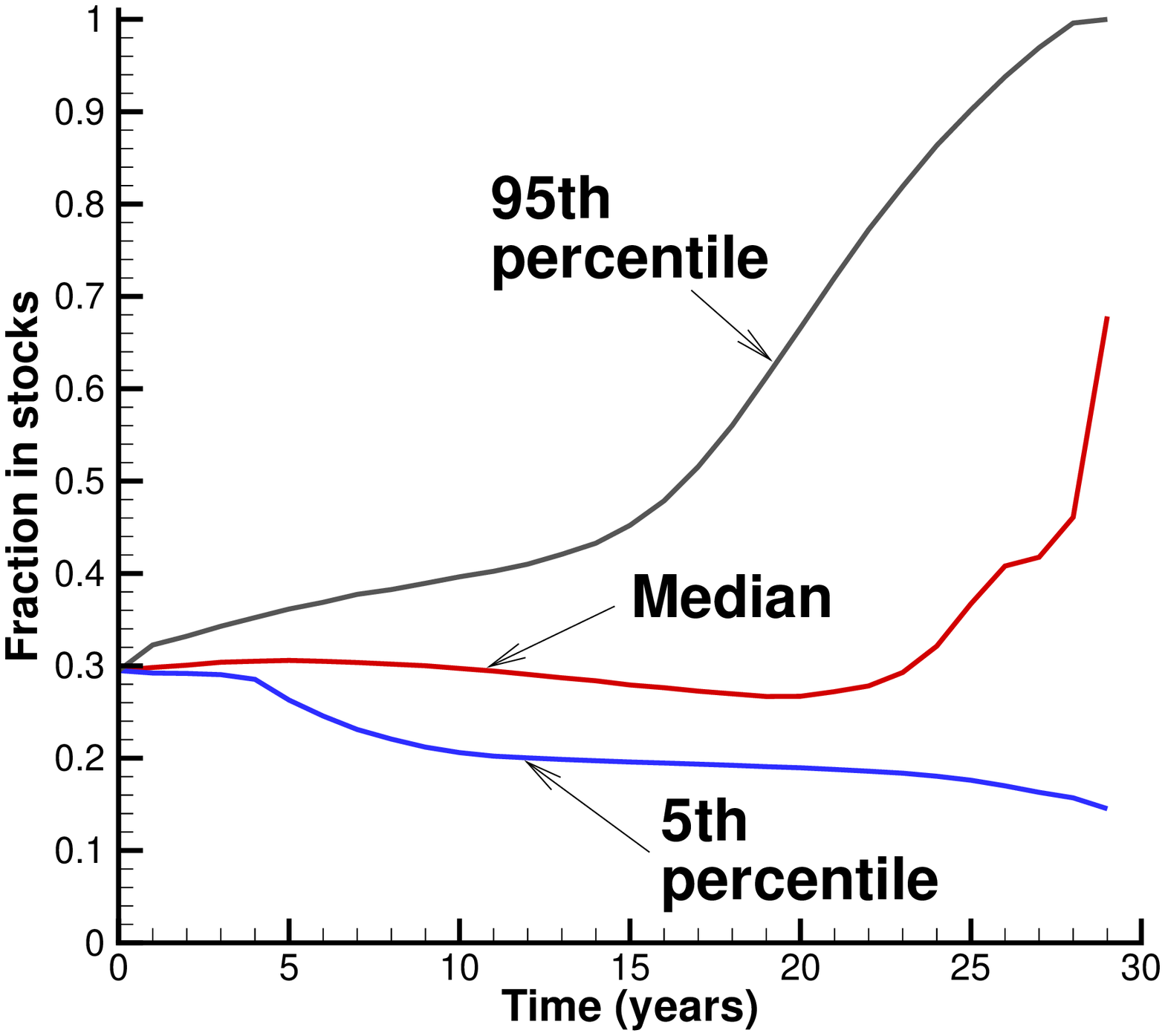}
\caption{Percentiles fraction in stocks}
\label{percentile_stocks_40_65}
\end{subfigure}
\begin{subfigure}[t]{.33\linewidth}
\centering
\includegraphics[width=\linewidth]{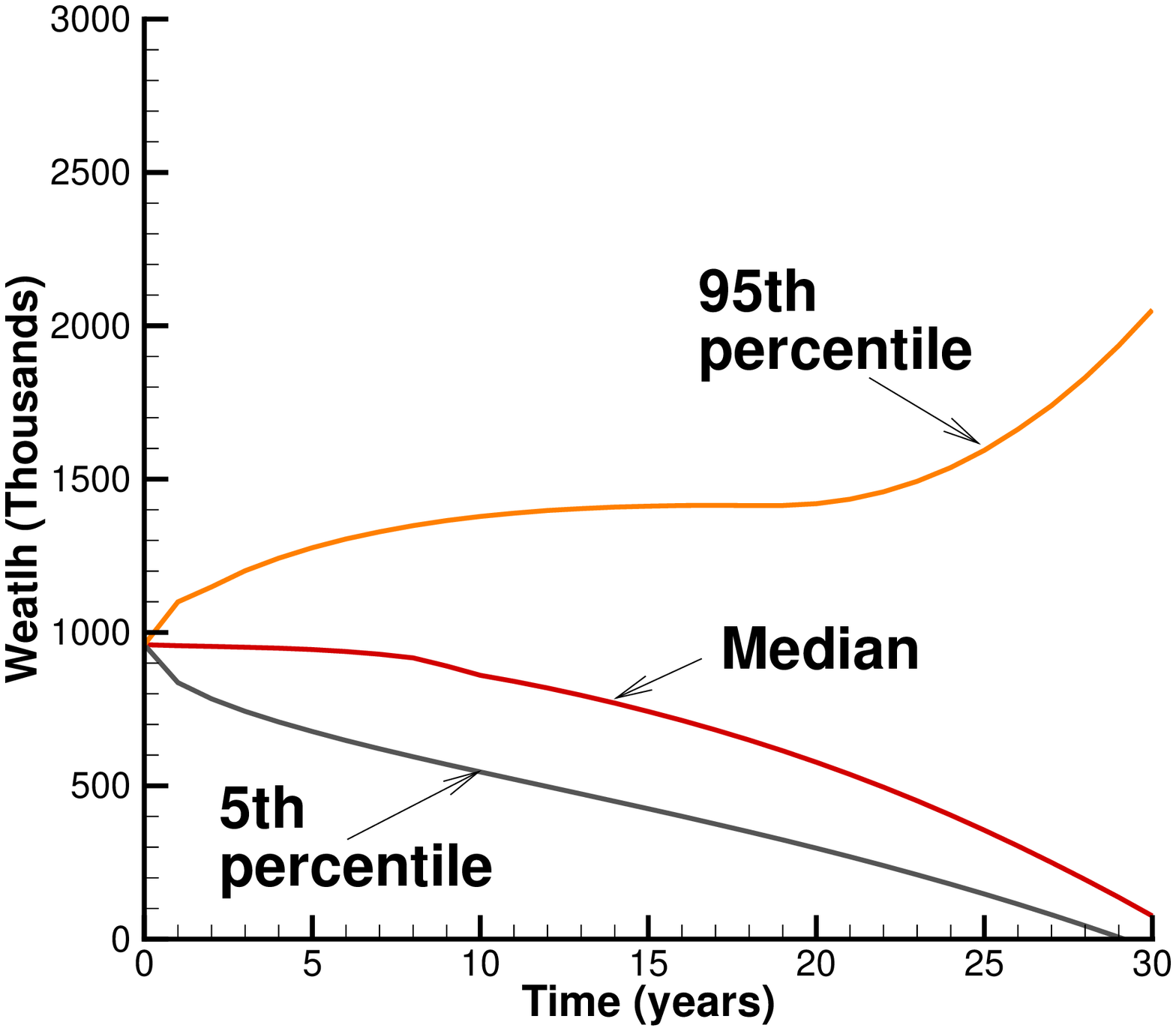}
\caption{Percentiles  wealth}
\label{percentiles_wealth_40_65}
\end{subfigure}
\begin{subfigure}[t]{.33\linewidth}
\centering
\includegraphics[width=\linewidth]{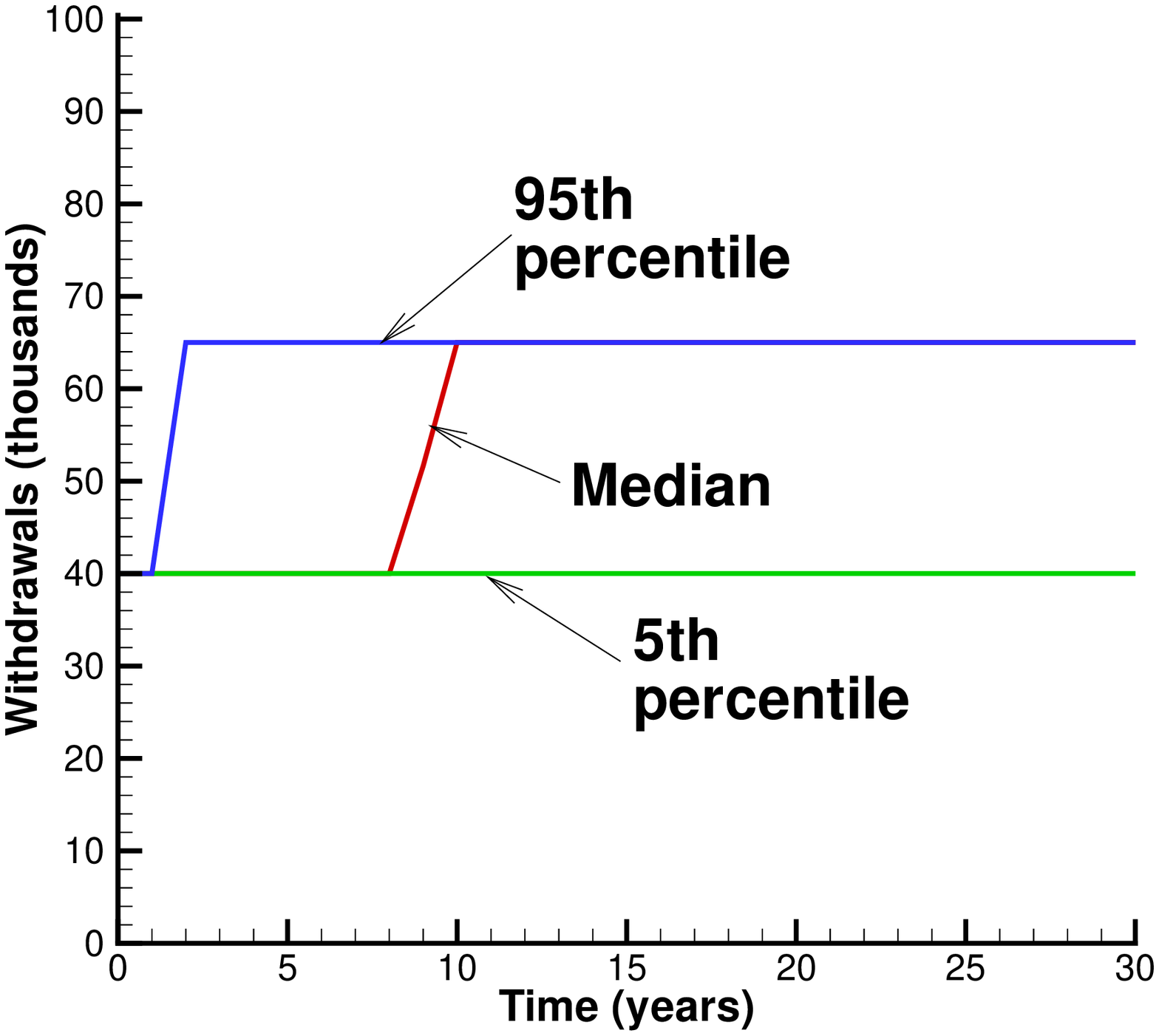}
\caption{Percentiles withdrawals}
\label{percentiles_q_40_65}
\end{subfigure}
}
\caption{Scenario in Table \ref{base_case_1}.
Optimal control computed from problem (\ref{PCEE_a}).
Parameters based on the  real CRSP index,
and real 10-year treasuries (see Table \ref{fit_params}).  Control computed and stored using
the PDE.
Scenario given in Table~\ref{base_case_1}.
Synthetic market, $2.56 \times 10^{6}$ MC simulations.
$q_{min} = 40, q_{\max} = 65$, $\kappa = 1.75$. $W^* = -28.2$.
$\epsilon = 10^{-6}$.
Units: thousands of dollars.
}
\label{percentiles_40_65}
\end{figure}

\begin{figure}[htb!]
\centerline{
\begin{subfigure}[t]{.4\linewidth}
\centering
\includegraphics[width=\linewidth]{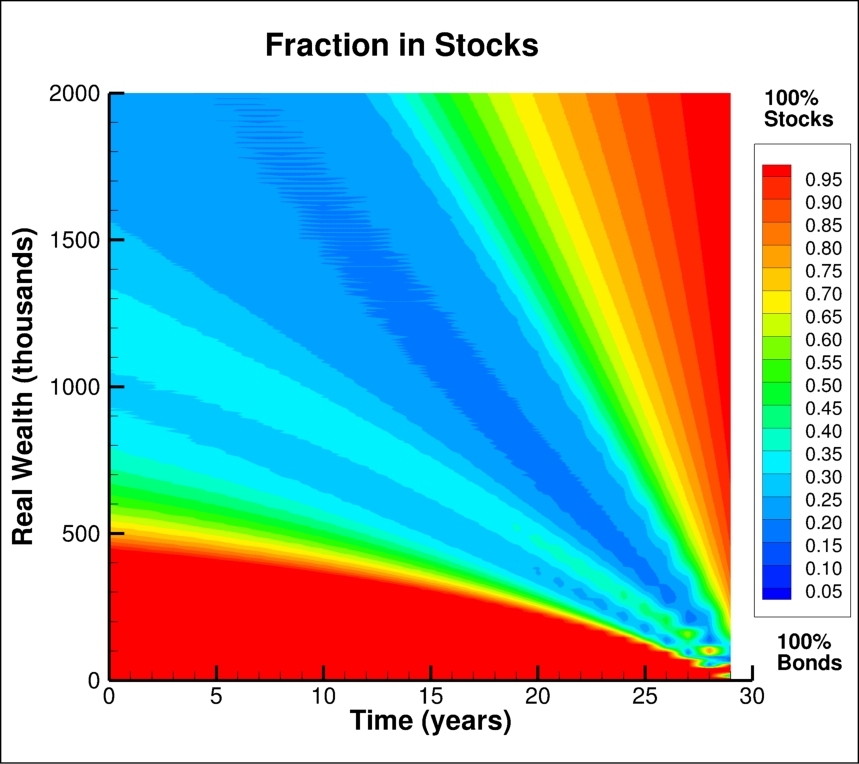}
\caption{Fraction in stocks}
\label{heat_tocks_40_65}
\end{subfigure}
\hspace{.05\linewidth}
\begin{subfigure}[t]{.4\linewidth}
\centering
\includegraphics[width=\linewidth]{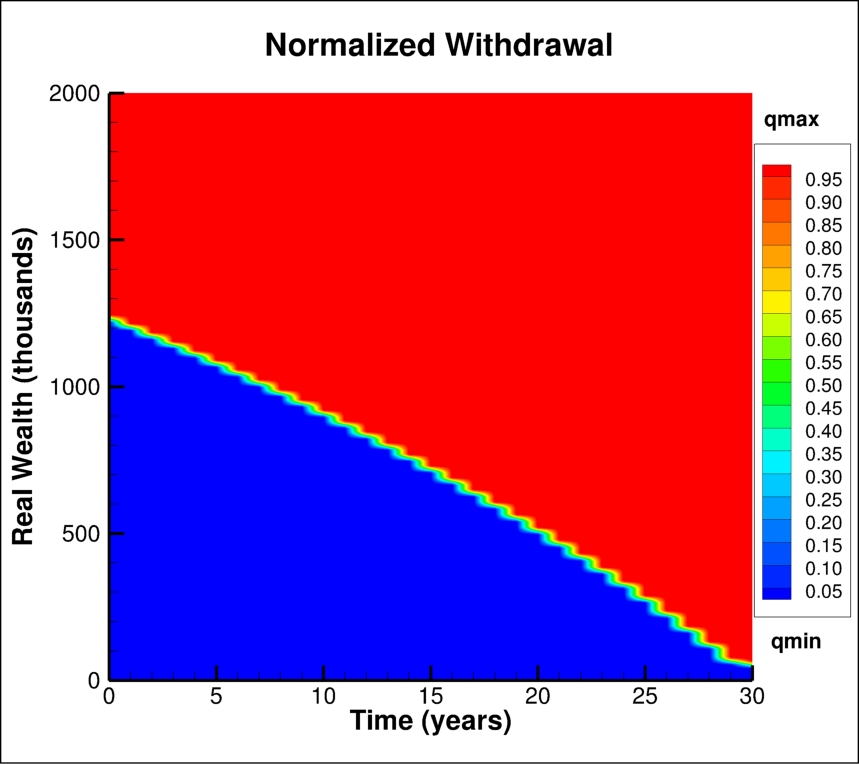}
\caption{Withdrawals}
\label{heat_withdrawal_40_65}
\end{subfigure}
}
\caption{
Heat map of controls: fraction in stocks and
withdrawals, computed from  problem (\ref{PCEE_a}).
Parameters based on the  real CRSP index,
and real 10-year treasuries (see Table \ref{fit_params}). 
Scenario given in Table~\ref{base_case_1}.
$q_{min} = 40, q_{\max} = 65$, $\kappa = 1.75$. $W^* = -28.2$.
$\epsilon = 10^{-6}$.
Normalized withdrawal $(q - q_{\min})/(q_{\max} - q_{\min})$.
Units: thousands of dollars.
\label{heat_map_40_65}}
\end{figure}

\clearpage

\section{Robustness check: historical market}
We compute and store the optimal controls from Problem (\ref{PCEE_a}), and then
use these controls in the bootstrapped historical market, as described in
Section \ref{boot_section}.
Table \ref{bootstrap_35_60} shows the effect of using different blocksizes
in the bootstrap simulations, compared to the synthetic market results.
The expected withdrawals are all very close, for all blocksizes. 
There is more variability in the ES results, but this  spread
is acceptable for practical purposes.  This indicates that the choice of blocksize will
not influence the qualitative results appreciably.  In the following,
we will report results using a blocksize of $.25$ years, which
is justified from Table \ref{auto_blocksize}.

The  detailed bootstrapped efficient frontiers (using the
controls computed in the synthetic market) are given in Tables 
\ref{boot_optimal_p_q_35_60} and \ref{boot_optimal_p_q_40_65}.
In Figure \ref{EW_ES_frontiers_boot}, we compare the EW-ES frontiers
computed for the case (i) controls computed in the synthetic market,
frontier computed in the synthetic market and (ii) controls
computed in the synthetic market, control tested in the
historical market. 
We can see that the synthetic market
frontiers are very close to the historical market frontiers.
This indicates that the controls computed in the synthetic
market are robust to uncertainty in the synthetic stochastic process
model calibrated to historical data.

\begin{figure}[htb!]
\centerline{%
\begin{subfigure}[t]{.40\linewidth}
\centering
\includegraphics[width=\linewidth]{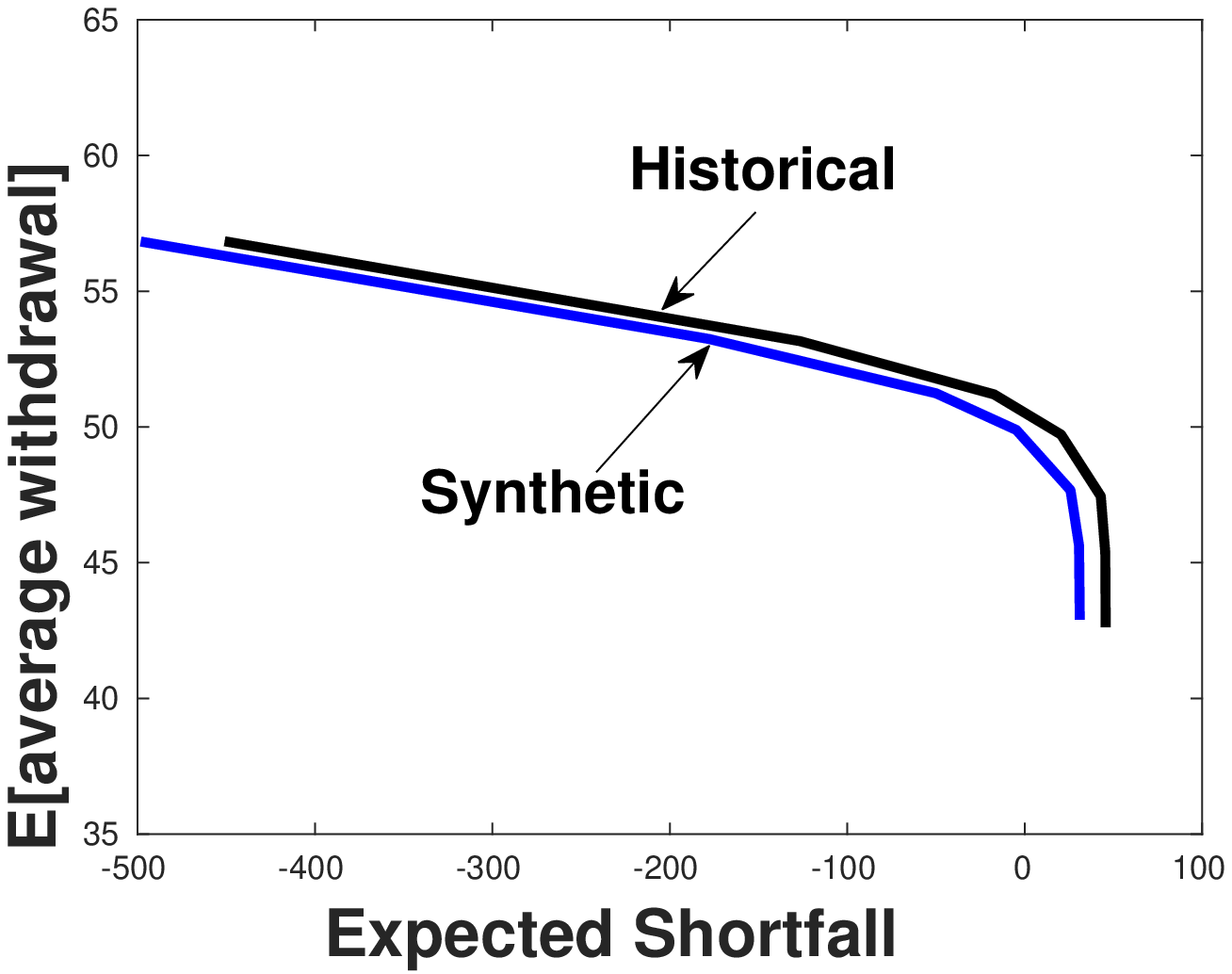}
\caption{$q_{\min}=35, q_{\max} = 60$.}
\label{frontier_35_60_boot}
\end{subfigure}
\begin{subfigure}[t]{.40\linewidth}
\centering
\includegraphics[width=\linewidth]{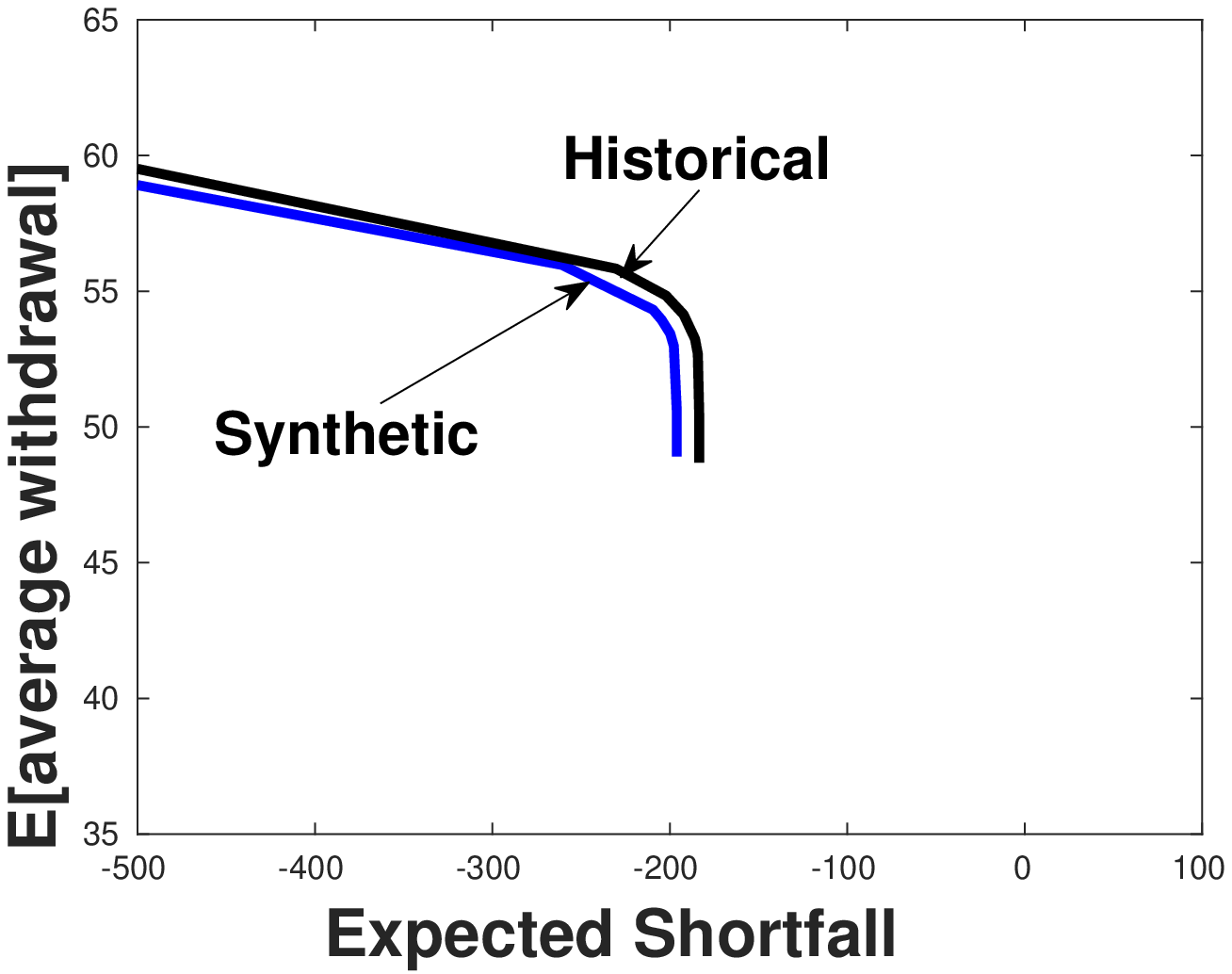}
\caption{$q_{\min}=40, q_{\max} = 65$.}
\label{frontier_40_65_boot}
\end{subfigure}
}
\caption{EW-ES frontiers, comparison of synthetic frontiers, and
frontier generated from (i) controls computed in the synthetic market
(ii) control tested in the historical (bootstrapped) market. Scenario in Table \ref{base_case_1}.
Parameters based on real CRSP index,
real 10-year US treasuries (see Table \ref{fit_params}).  Control computed and stored,
historical frontier computed using
$10^5$ bootstrap resampled simulations, blocksize $0.25$ years.
Historical data in range 1926:1-2019:12.
Units: thousands of dollars.
}
\label{EW_ES_frontiers_boot}
\end{figure}

\section{Discussion}
Adding a variable withdrawal strategy, coupled with optimal asset allocation,
can dramatically improve the expected average withdrawal, compared with
a constant withdrawal strategy.  If the minimum withdrawal of the variable
strategy is set equal to the constant withdrawal strategy, then this result
still holds, requiring only a small increase in risk, as measured by expected
shortfall (ES).

At first sight, this result is almost too good to be true.  However, this is
easily explainable, due to two effects.  
\begin{itemize}
   \item The median final wealth of the variable withdrawal strategy is much lower
         than the constant withdrawal policy.  Hence, the variable withdrawal strategy
         is much more efficient in disbursing cash to the retiree over the investment
         horizon, while keeping the overall risk almost unchanged.\footnote{{\em{``If we have a good year,
         we take a trip to China,...if we have a bad year, we stay home and play canasta.''}}
         retired professor Peter Ponzo, discussing his DC plan withdrawal strategy
         {\url{https://www.theglobeandmail.com/report-on-business/math-prof-tests-investing-formulas-strategies/article22397218/}}}

   \item Due to the  quasi-bang-bang control for the variable withdrawal strategy,
         the median optimal policy is to withdraw at the minimum rate for the first few years,
         followed by withdrawing at the maximum rate for 20-25 years.  This avoids large
         withdrawals in the early years, ameliorating sequence of return risk,
         with the benefits to be gained in later years.
\end{itemize}

The downside of this strategy is that, although the average withdrawal is signficently 
improved, the first few years after retirement typically have the smallest (minimum)
withdrawals.  This may not be desirable, if retirees are most active at this time,
and may wish to have larger incomes.

There are several ways to move spending earlier, but these all come at some cost
in terms of EW-ES efficiency.  Recall that all quantities in this paper are real,
hence we are always preserving real spending power.  However, we could add a real
discounting multiplier to our measure of reward.  This would change equation (\ref{EW_def})
to
\begin{eqnarray}
      {\text{EW}}( X_0^-, t_0^-) = E_{\mathcal{P}_0}^{X_0^+, t_0^+}
                                 \biggl[
                                    \sum_{i=0}^{i=M} e^{- \beta t_i} q_i
                                 \biggr] ~,  \label{EW_def_discounted}
\end{eqnarray}
where $\beta > 0$ is a discounting parameter. We experimented with this approach, and it
did tend to move more spending earlier, but at the expense of more risk.  In fact,
the results using a discounting factor were similar to simply decreasing the
scalarization parameter $\kappa$ in equation (\ref{objective_overview}).  
The withdrawal
control in this case was also quasi-bang-bang.  
Adding a discount factor does not change the bang-bang nature of the withdrawal
control, at least in the continuous
withdrawal limit.  This can be be verified by adding a discounting factor to equation (\ref{expanded_1_cont}).

In order to produce a withdrawal control which is more gradual (not bang-bang), we need
to add a nonlinearity to the measure of reward.  Let $\mathcal{U}(\cdot)$ be a utility function,
then our measure of reward could be
\begin{eqnarray}
   {\text{EW}}( X_0^-, t_0^-) = E_{\mathcal{P}_0}^{X_0^+, t_0^+}
                                 \biggl[
                                    \sum_{i=0}^{i=M}  \mathcal{U} ( q_i )
                                 \biggr] ~.  \label{EW_def_utility}
\end{eqnarray}
We experimented with various utility functions (e.g. $\log$, power law), and this did have the
effect of producing smoother controls as a function of wealth.  However, this came at the
cost of poor EW-ES efficiency.  Recall that our initial objective in this work was to provide
the retiree with fixed minimum cash flows, with small risk,  while maximizing total
withdrawals.   Using a nonlinear utility function
would conflict with this criteria.  We leave exploration of the use of a nonlinear utility 
in the reward function as a topic for future work.

\section{Conclusions}
Our objective in this work was to provide a retiree with minimum fixed cash flows over
a long time horizon, with high probability of expected average withdrawals being
signficently larger than the minimum withdrawal.  In addition, we control the risk
of this strategy as measured by expected shortfall.  

The optimal controls consisted of
a variable withdrawal rate (with minimum and maximum constraints) and the asset
allocation strategy.  Allowing a variable withdrawal strategy (compared to a fixed
withdrawal) dramatically improved the expected average withdrawals, at the expense
of a very small increase in expected shortfall risk.  However, the early withdrawals
were (with high probability) at the minimum level, and larger withdrawals were achieved
later on in life.

Note that the controls were computed in the synthetic market, i.e. a market based
on a parametric stochastic process model calibrated to data over the 1926:1-2019:12 
period.  However, bootstrap resampling tests showed that this strategy is
robust to model and parameter uncertainty.

An intriguing application of this research is the following.  In many countries (Canada
in particular), there is a reward, in terms of increased cash flows, if the retiree
delays receiving government benefits, until later ages (e.g. 70 in Canada).
The common advice is to delay receiving government benefits, and
offset this with larger drawdowns from the DC account
in the early years of retirement.  The argument here is that government benefits
are indexed and certain, compared with uncertain investment cash flows.

However, our results indicate that fairly small reductions in withdrawals from
the DC account in early years result in much larger withdrawals in later years,
with a high probability.  Hence a better strategy may be to take some government
benefits earlier, allowing smaller withdrawals from the DC account in
early years (reducing sequence of return risk).  The smaller government benefits in later years will (again, with
high probability) be offset by these much larger withdrawals from the DC account.
Of course, although this strategy has a high probability of success, it is not
risk-free.

\section{Acknowledgements}
P. A. Forsyth's work was supported by the Natural Sciences and Engineering Research Council of
Canada (NSERC) grant RGPIN-2017-03760.  

\section{Conflicts of interest}
The author has no conflicts of interest to report.

\clearpage
\appendix
\section*{Appendix}

\section{Detailed efficient frontiers: synthetic market}
\label{Detailed_synthetic}
Tables \ref{optimal_p_q} and \ref{optimal_p_q_2} give the detailed
results used to construct Figure \ref{EW_ES_frontiers}.

\begin{table}[hbt!]
\begin{center}
\begin{tabular}{ccccc} \toprule
$\kappa$ & ES (5\%) & $E[ \sum_i q_i]/(M+1)$  & $Median[W_T]$ & $\sum_i Median(p_i) /M$     \\ \midrule
  0.05 &  -498.2   & 56.83               &  119.7       &  .430 \\
 0.2   & -177.9 & 53.24                   & 324.9        & .405 \\
 0.5 &   -50.86  & 51.33                  & 368.2        & .363\\
 1.0 &   -4.730        &49.89                   & 406.3          & .331\\
 5.0 &     25.79    & 47.67                    &  451.8         & .282\\
 50.0 &    30.62    & 45.63                    & 524.6          & .259\\
 5000.0  &   31.02   &42.90                     & 661.8        & .252\\
\bottomrule
\end{tabular}
\caption{Synthetic market results for optimal strategies,  assuming
the scenario given in Table~\ref{base_case_1}. Stock index: real capitalization weighted CRSP stocks;
bond index: ten year treasuries.  Parameters from Table \ref{fit_params}.
Units: thousands of dollars. Statistics
based on $2.56 \times 10^6$ Monte Carlo simulation runs.
Control is computed using the Algorithm in Section \ref{algo_section}, stored, and then used
in the Monte Carlo simulations.
$q_{\min} = 35.0$, $q_{\max} = 60$.
$(M+1)$ is the number of withdrawals.
$M$ is the number of rebalancing dates.
$\epsilon = 10^{-6}$.
\label{optimal_p_q}
}
\end{center}
\end{table}

\begin{table}[hbt!]
\begin{center}
\begin{tabular}{ccccc} \toprule
$\kappa$ & ES (5\%) & $E[ \sum_i q_i]/(M+1)$  & $Median[W_T]$ & $\sum_i Median(p_i) /M$     \\ \midrule
  0.1    & -587.3    & 59.98        & 46.36  & .455 \\
  0.5    & -260.8   &  55.97    &   75.06    & .418 \\
  1.0 &   -237.2    & 55.00    &   73.61    & .366 \\
  1.75&    -209.5   & 54.32     &  75.49  & .341 \\
  2.5  &  - 204.7   &  53.95     & 78.56    & .331  \\
  5.0 &   -199.8    &  53.44     &  78.31   &  .314 \\
 10.0 &  -197.8    &  52.98      &  91.68   &  .303 \\
 $10^{3}$ &  -196.2  &  50.63      &  202.1  & .285 \\
 $10^{5}$ & -196.1 & 48.91       &  316.0  &  .303 \\
\bottomrule
\end{tabular}
\caption{Synthetic market results for optimal strategies,  assuming
the scenario given in Table~\ref{base_case_1}. Stock index: real capitalization weighted CRSP stocks;
bond index: ten year treasuries.  Parameters from Table \ref{fit_params}.
Units: thousands of dollars. Statistics
based on $2.56 \times 10^6$ Monte Carlo simulation runs.
Control is computed using the Algorithm in Section \ref{algo_section}, stored, and then used
in the Monte Carlo simulations.
$q_{\min} = 40.0$, $q_{\max} = 65$.
$(M+1)$ is the number of withdrawals.
$M$ is the number of rebalancing dates.
$\epsilon = 10^{-6}$.
\label{optimal_p_q_2}
}
\end{center}
\end{table}

\clearpage

\section{Effect of blocksize: stationary block bootstrap resampling}
\label{block_app}
Table \ref{bootstrap_35_60} shows the effect of blocksize on the bootstrap resampling
algorithm.

\begin{table}[hbt!]
\begin{center}
\begin{tabular}{ccccc} \toprule
$\kappa$ & ES (5\%) & $E[ \sum_i q_i]/(M+1)$  & $Median[W_T]$ & $\sum_i Median(p_i) /M$     \\ \midrule
\multicolumn{5}{c}{Synthetic Market} \\ \midrule
  0.5 &       -50.86  & 51.33                  & 368.2        & .363\\
  \midrule
\multicolumn{5}{c}{Historical Market: $\hat{b} = 0.25$ years.} \\ \midrule
         &    -17.28      & 51.19     & 340.6  & .360 \\ \midrule
\multicolumn{5}{c}{Historical Market: $\hat{b} = 0.5$ years.} \\ \midrule
         &  -40.63        &51.19     & 343.1  & .361   \\ \midrule
\multicolumn{5}{c}{Historical Market: $\hat{b} = 1$ years.} \\ \midrule
         &  -34.13        & 51.23    & 342.9  &   .361  \\ \midrule
 \multicolumn{5}{c}{Synthetic Market} \\ \midrule
1.0    & 4.730    & 49.89    &   406.3    & .331 \\ \midrule
\multicolumn{5}{c}{Historical Market: $\hat{b} = 0.25$ years.} \\ \midrule
       &  20.47  & 49.72   & 381.0   & .330 \\ \midrule
\multicolumn{5}{c}{Historical Market: $\hat{b} = 0.5$ years.} \\ \midrule
       & -5.10    & 49.72   & 383.6   &  .331 \\ \midrule
\multicolumn{5}{c}{Historical Market: $\hat{b} = 1$ years.} \\ \midrule
       & -0.84  &  49.74   & 384.4   & .331    \\
\bottomrule
\end{tabular}
\caption{Historical market results for optimal strategy, $q_{\min} = 35, q_{\max} = 60$.
The scenario is given in Table~\ref{base_case_1}. Stock index: real capitalization weighted CRSP stocks;
bond index: 10 year treasuries.
Historical data in range 1926:1-2019:12.
Units: thousands of dollars. Statistics
based on $10^5$ bootstrap simulations.
Control is computed using the algorithm in Section \ref{algo_section}, stored, and then used
in the bootstrap resampling tests.
$(M+1)$ is the number of withdrawals. $M$ is the number of rebalancing dates.
$\epsilon = 10^{-6}$.
\label{bootstrap_35_60}
}
\end{center}
\end{table}

\section{Bootstrapped frontiers}
Tables \ref{boot_optimal_p_q_35_60} and \ref{boot_optimal_p_q_40_65} show the detailed 
results for the EW-ES frontiers.  The controls were computed in the synthetic market,
and tested in the historical market.

\begin{table}[hbt!]
\begin{center}
\begin{tabular}{ccccc} \toprule
$\kappa$ & ES (5\%) & $E[ \sum_i q_i]/(M+1)$  & $Median[W_T]$ & $\sum_i Median(p_i) /M$     \\ \midrule
  0.05 &  -450.9    &  56.84    &  90.31   &   .424\\
 0.2   & -126.7      &  53.16     & 294.7  & .402\\
 0.5 &  -17.28        &   51.20    &  340.6  & .360\\
 1.0 &   20.47    &  49.72   &     381.0     & .330\\
 5.0 &   42.86    &   47.45  &    430.3  &    .282\\
 50.0 &   45.39   &   45.38      &  502.5   & .258\\
 5000.0 &  45.64   &  42.62       & 638.9     & .252 \\
 \midrule
 \multicolumn{5}{c}{$q_{\max} = q_{\min} = 35$} \\ \midrule
N/A     & 45.63    &  35.0     & 920.0   & .269  \\
\bottomrule
\end{tabular}
\caption{Control computed in the synthetic market,  assuming
the scenario given in Table~\ref{base_case_1}. Stock index: real capitalization weighted CRSP stocks;
bond index: ten year treasuries.  Parameters from Table \ref{fit_params}.
Units: thousands of dollars. Statistics
based on $10^5$ bootstrap resampling of the historical data.
Historical data in range 1926:1-2019:12.
Expected blocksize $\hat{b} = .25 $ years.
$q_{\min} = 35.0$, $q_{\max} = 60$.
$(M+1)$ is the number of withdrawals.
$M$ is the number of rebalancing dates.
\label{boot_optimal_p_q_35_60}
}
\end{center}
\end{table}

\begin{table}[hbt!]
\begin{center}
\begin{tabular}{ccccc} \toprule
$\kappa$ & ES (5\%) & $E[ \sum_i q_i]/(M+1)$  & $Median[W_T]$ & $\sum_i Median(p_i) /M$     \\ \midrule
  0.1    & -531.8   & 59.94   & 25.99  & .453   \\
  0.5    &  -230.3   & 55.83  & 53.89   &  .412  \\
  1.0 &     -202.2   & 54.84  &  54.30 & .363 \\
  1.75  & -192.3    &  54.14   & 57.13  & .337 \\
  5.0 &  -185.8     &  53.24   & 61.79   & .310\\
 10.0 &   -184.3    & 52.69    & 75.40    &  .299 \\
 1000.0 & -183.5    &  50.43   &  186.3   & .283\\
 $10^{5}$ & -183.5  &   48.69  &  296.9   &  .300\\ \midrule
 \multicolumn{5}{c}{$q_{\max} = q_{\min} = 40$} \\ \midrule
   N/A  &  -183.5   & 40.0     & 676.7    &  .340  \\
\bottomrule
\end{tabular}
\caption{Control computed in the synthetic market,  assuming
the scenario given in Table~\ref{base_case_1}. Stock index: real capitalization weighted CRSP stocks;
bond index: ten year treasuries.  Parameters from Table \ref{fit_params}.
Units: thousands of dollars. Statistics
based on $10^5$ bootstrap resampling of the historical data.
Historical data in range 1926:1-2019:12.
Expected blocksize $\hat{b} = .25 $ years.
$q_{\min} = 40.0$, $q_{\max} = 65$.
$(M+1)$ is the number of withdrawals.
$M$ is the number of rebalancing dates.
\label{boot_optimal_p_q_40_65}
}
\end{center}
\end{table}

\clearpage

\bibliographystyle{chicago}
\bibliography{paper}

\end{document}